\documentclass[a4paper,11pt]{article}

\usepackage{fullpage,amsthm,color}
\usepackage{titlesec,authblk}

\definecolor{myurlcolor}{rgb}{0,0,0.4}
\definecolor{mycitecolor}{rgb}{0,0.5,0}
\definecolor{myrefcolor}{rgb}{0.5,0,0}
\usepackage{hyperref}
\hypersetup{colorlinks,
linkcolor=myrefcolor,
citecolor=mycitecolor,
urlcolor=myurlcolor}

\usepackage{amssymb, amsmath, amsthm, verbatim}
\usepackage[mathscr]{euscript}
\usepackage[draft]{fixme}
\usepackage{tikz}
\usepackage{subfigure}
\usepackage{multicol}
\usepackage{amsmath,bbm}%
\usepackage{graphicx}
\usepackage{amsfonts}
\usepackage{amssymb}
\usepackage{amsmath, amssymb, amsthm,verbatim,graphicx,bbm}
\usepackage{color,xcolor,longtable}
\usepackage{times}

\usepackage{changes}

\usepackage{hyperref}

\usetikzlibrary{patterns}
\newcommand{\beq}[0]{\begin{equation}}
\newcommand{\eeq}[0]{\end{equation}}

\newcommand{\one}{\leavevmode\hbox{\small1\normalsize\kern-.33em1}}

\def\be{\begin{equation}}
\def\ee{\end{equation}}
\def\ben{\begin{eqnarray}}
\def\een{\end{eqnarray}}
\def\eea{\end{array}}
\def\bea{\begin{array}}

\newcommand{\Tr}[1]{\mathrm{Tr}#1}
\newcommand{\bei}{\begin{itemize}}
\newcommand{\eei}{\end{itemize}}
\newcommand{\ket}[1]{|#1\rangle}
\newcommand{\bra}[1]{\langle#1|}

\newcommand{\proj}[1]{\ket{#1}\!\bra{#1}}

\renewcommand{\emph}[1]{\textbf{#1}}

\theoremstyle{plain}
\newtheorem{thm}{Theorem}%

\theoremstyle{definition}

\theoremstyle{remark}

\titleformat*{\section}{\large\bfseries}

\begin{document}

\title{Bell inequalities tailored to the Greenberger-Horne-Zeilinger states of arbitrary local dimension}
\author{R. Augusiak$^1$, A. Salavrakos$^2$, J. Tura$^{3}$, A. Ac\'in$^{2,4}$}

\affil{
$^1$Center for Theoretical Physics, Polish Academy of Sciences,\\ Aleja Lotnik\'ow 32/46, 02-668 Warsaw, Poland\\
$^2$ICFO-Institut de Ciencies Fotoniques, The Barcelona Institute of Science and Technology, \\08860 Castelldefels (Barcelona), Spain\\
$^3$Max-Planck-Institut f\"ur Quantenoptik, Hans-Kopfermann-Stra{\ss}e 1, 85748 Garching, Germany\\
$^4$ICREA, Pg. Lluis Companys 23, 08010 Barcelona, Spain\\
}
\renewcommand\Affilfont{\itshape\small}

\date{}

\maketitle

\begin{abstract}

In device-independent quantum information processing Bell inequalities are not only 
used as detectors of nonlocality, but also as certificates of relevant quantum properties.  
In order for these certificates to work, one very often needs Bell inequalities that are maximally violated by specific quantum states. 
Recently, in [A. Salavrakos \textit{et al.}, \href{https://journals.aps.org/prl/abstract/10.1103/PhysRevLett.119.040402}{Phys. Rev. Lett. \textbf{119}, 040402 (2017)}] a general class of Bell inequalities, with arbitrary numbers of measurements and outcomes, has been designed, which are maximally violated by the maximally entangled states of two quantum systems of arbitrary dimension. In this work, we generalize these results to the multipartite scenario and obtain a general class of Bell inequalities maximally violated by the Greenberger-Horne-Zeilinger states of any number of parties and any local dimension. We then derive analytically their maximal quantum and nonsignaling values. We also obtain analytically the bound for detecting genuine nonlocality and compute the fully local bound for a few exemplary cases. Moreover, we consider the question of adapting this class of inequalities to partially entangled GHZ-like states for some special cases of low dimension and small number of parties. Through numerical methods, we find classes of inequalities maximally violated by these partially entangled states.
\end{abstract}

\section{Introduction}

Bell inequalities \cite{Bell} have traditionally been used as witnesses of nonlocality in composite quantum systems, but with the advent of device-independent quantum information processing they gained a completely new role as certificates of relevant quantum properties. It is nowadays a well-established fact that the violation of Bell inequalities not only certifies the presence of entanglement in a device-independent way, but it can also certify e.g. that true randomness has been generated in the process of measuring a quantum system \cite{randomness1}. Among their certification properties, Bell inequalities may serve as device-independent witnesses of the minimal Hilbert space dimension of the underlying quantum system \cite{dimwit}. The maximum exponent of their certification power is known as self-testing \cite{MayersYao}, which allows one to conclude the state and measurements performed solely from the observation of the maximal violation of certain Bell inequalities (see, e.g., Refs. \cite{selftesting}).

In many of these device-independent applications, in particular in randomness certification \cite{randomness1} or self-testing \cite{MayersYao,selftesting}, one needs Bell inequalities whose maximal quantum values are known along with the quantum realisation (that is, a quantum state and quantum measurements) achieving them. This is not an easy task in general, because, phrasing alternatively, one needs Bell inequalities maximally violated by specific quantum states and/or specific quantum measurements. While many constructions of Bell inequalities, both in the bipartite and multipartite cases (see, e.g., Refs. \cite{CHSH,tilted,CGLMP,BM05,BKP,VP08,Ji08,Liang09,Wiesiek,GraphBell,SLK06,Lim10,Aolita}), have been proposed to date, the quantum realisation maximally violating these inequalities is characterized only for a proper subset of them, and most of these inequalities involve two-outcome measurements. 
In the bipartite case these are for instance: the Clauser-Horne-Shimony-Holt (CHSH) Bell inequality \cite{CHSH}, which is maximally violated by the maximally entangled state of two qubits, its generalization, called the \textit{tilted} CHSH \cite{tilted}, which is maximally violated by any partially entangled two-qubit state, and the generalizations of the CHSH Bell inequality to inequalities maximally violated by the maximally entangled state of arbitrary local dimension and various measurements \cite{our,Jed,Andrea}, devised only recently. 
Moving to the multipartite case, examples of Bell inequalities for which the realization of the maximal quantum violation is known are: the Mermin Bell inequality \cite{Mermin}, the class of Bell inequalities maximally violated by the multiqubit graph states \cite{GraphBell} (see also Ref. \cite{ourGraph} for the recent alternative construction), or a class of two-setting Bell inequalities introduced in Ref. \cite{SLK06} and tailored to the $N$-partite Greenberger-Horne-Zeilinger states of arbitrary local dimension 
\begin{equation}\label{GHZ}
|\mathrm{GHZ}_{N,d}\rangle=\frac{1}{\sqrt{d}}\sum_{i=0}^{d-1}\ket{i}^{\otimes N},
\end{equation}
for which the maximal quantum violation was determined only later in 
Ref. \cite{LCL07}.

The main aim of this work is to design a new family of Bell inequalities for
which one can efficiently determine the maximal quantum violation along with the 
quantum realisation achieving it. We provide a general class of multipartite
Bell inequalities valid for any number of measurements and outcomes whose maximal quantum 
violation is attained by the GHZ state of $N$ qudits (\ref{GHZ}). To this end, we exploit and, at the same time extend to the multipartite scenario, the approach in Ref. \cite{our} to construct Bell inequalities for the maximally entangled state of two qudits $|\mathrm{GHZ}_{2,d}\rangle$. Noticeably, this approach exploits the properties of the quantum state and measurements to derive Bell inequalities, rather than the geometry of the set of local correlations. We then characterize the obtained inequalities: (i) first, we compute their maximal classical values for the simplest multipartite scenarios 
(note that for $N=2$ these values were already computed analytically in Ref. \cite{our}), (ii) we detemine their maximal quantum value by finding a sum of squares decomposition of the corresponding Bell operator, and (iii) compute their maximal nonsignaling values. 
In the spirit of Refs. \cite{our, Science}, we finally discuss generalizations of our Bell inequalities to certain partially entangled multipartite states.

Noticeably, our class of Bell inequalities reproduces the two-setting Bell inequalities
introduced in Ref. \cite{SLK06} and later studied in Ref. \cite{LCL07}. On the other hand, it belongs to a broader class of multipartite Bell functionals considered in Ref. \cite{Bancal}.
Nevertheless, this last work, although it reproduces notable inequalities such as the ones presented in Ref. \cite{AGCA}, it does not single out the class of Bell inequalities nor the properties we provide in this work.
Moreover, here we provide a different approach to compute the maximal quantum violation that is based on the sum-of-squares decomposition of the Bell operator.

The manuscript is organized as follows. In Sec. \ref{sec:preliminaries}
we recall all the relevant notions for further considerations. 
In Sec. \ref{sec:construction} we derive our family of Bell inequalities, whereas in 
Sec. \ref{sec:characterization} we characterize them by providing their maximal
classical (numerically, in the simplest scenarios), quantum and nonsignaling values.
In Sec. \ref{sec:generalizations} we put forward possible generalizations of our construction
to partially entangled GHZ multiqudit states, and we present our conclusions in Sec. \ref{sec:conclusion}.

\section{Preliminaries}
\label{sec:preliminaries}

\textit{Bell scenario and correlations.}
Let us consider a Bell scenario in which $N$ distant parties $A_1,\ldots,A_N$
share some physical system. In each round of the experiment, each party $A_i$ performs
one of $m$ measurements on their share of this system, 
and each measurement yields one of $d$ outcomes. We label the measurement choices and outcomes of party $A_i$ by $x_i\in\{1,\ldots,m\}$ and $a_i\in\{0,\ldots,d-1\}$, respectively, while $A_{i,x_i}$ denotes the implemented measurement. Such measurements lead to correlations that are described by a collection of conditional probability distributions  
\begin{equation}\label{collection}
\{p(a_1,\ldots,a_N|x_1,\ldots,x_N)\}_{a_1,\ldots,a_N;x_1,\ldots,x_N},
\end{equation}
in which $p(\boldsymbol{a}|\boldsymbol{x}):= p(a_1,\ldots,a_N|x_1,\ldots,x_N)$ stands for the probability of obtaining outcomes $\boldsymbol{a}:= (a_1,\ldots, a_N)$ upon performing measurements $\boldsymbol{x}:= (x_1,\ldots, x_N)$ by the parties. These probabilities are typically ordered into a vector 
\begin{equation}
\vec{p}=\{p(a_1,\ldots,a_N|x_1,\ldots,x_N)\}_{a_1,\ldots,a_N;x_1,\ldots,x_N}\in \mathbbm{R}^{(md)^N}.
\end{equation}
By slightly abusing terminology we also call the collection (\ref{collection}) correlations.

Now, the set of allowed vectors $\vec{p}$ varies depending on the physical principle 
they obey. First, let us consider correlations that satisfy the no-signaling principle
which prohibits faster-than-light communication between parties. Mathematically, this is equivalent to saying that the conditional probabilities $p(\mathbf{a}|\mathbf{x})$
satisfy the following set of linear constraints
\begin{equation}
\sum_{a_i}p(a_1,\ldots,a_i,\ldots,a_N|x_1,\ldots,x_i,\ldots,x_N)=
\sum_{a_i}p(a_1,\ldots,a_i,\ldots,a_N|x_1,\ldots,x_i',\ldots,x_N)
\label{eq:ns}
\end{equation}
for all $x_i,x_i'$ and $a_1,\ldots,a_{i-1},a_{i+1},\ldots,a_N$ and $x_1,\ldots,x_{i-1},x_{i+1},\ldots,x_N$ and all $i$. Correlations obeying the no-signaling principle form a convex polytope that we will denote by $\mathcal{N}_{N,m,d}$.

The polytope $\mathcal{N}_{N,m,d}$ contains the set of quantum correlations, which are those that can be represented as 
\begin{equation}\label{quantum}
p(\mathbf{a}|\mathbf{x})=\Tr[\rho_N (M_{1,x_1}^{a_1}\otimes\ldots\otimes M_{N,x_N}^{a_N})]
\end{equation}
for some $N$-partite quantum state $\rho_N$ of generally unconstrained dimension and local positive semi-definite measurement operators $M_{i,x_i}^{a_i}$ that define the $x_i$th measurement (with outcome $a_i$) performed by party $A_i$. Since $M_{i,x_i}^{a_i}$ are positive-operator valued measure (POVM) elements, they form a resolution of the identity: $\sum_{a_i}M_{i,x_i}^{a_i}=\mathbbm{1}$.
Similarly to the nonsignaling set, the quantum set $\mathcal{Q}_{N,m,d}$ is also convex, however, it is not a polytope. Moreover, as shown in Ref. \cite{PR}, $\mathcal{Q}_{N,m,d}$, is a proper subset of $\mathcal{N}_{N,m,d}$ as there exist correlations obeying the no-signaling principle that do not have the above quantum realization (\ref{quantum}).

Finally, the set of correlations that admit the local hidden variable (LHV) models is formed by those $\vec{p}$ for which every $p(\boldsymbol{a}|\boldsymbol{x})$ can be written as a convex combination of 
product deterministic correlations, that is,
\begin{equation}
p(\mathbf{a}|\mathbf{x})=\sum_{\lambda}p(\lambda) p_{A_1}(a_1|x_1,\lambda)\cdot\ldots\cdot p_{A_N}(a_N|x_N,\lambda).
\end{equation}
Here $\lambda$ is some classical information (which can also be interpreted as a hidden variable or shared randomness) and $p(a_i|x_i,\lambda)\in\{0,1\}$ for all 
$a_i,x_i$ and $\lambda$. In what follows we will also refer to correlations admitting the above representation as to local or classical. Likewise the nonsignalling set, the local set is a polytope whose vertices are product of deterministic correlations, i.e., $p(\mathbf{a}|\mathbf{x})=p_{A_1}(a_1|x_1)\cdot\ldots\cdot p_{A_N}(a_N|x_N)$ with each $p_{A_i}(a_i|x_i)\in \{0,1\}$.

It is important to notice that $\mathcal{L}_{N,m,d}$ is a proper subset of $\mathcal{Q}_{N,m,d}$ and Bell was the first to prove that not all quantum correlations admit an LHV model. To this end, he used the concept of Bell inequalities---linear inequalities constraining the local polytope $\mathcal{L}_{N,m,d}$ that take the general form
\begin{equation}
I:=\sum_{\mathbf{a},\mathbf{x}}T_{\mathbf{a},\mathbf{x}}\,p(\mathbf{a}|\mathbf{x})\leq \beta_{\mathcal{L}},
\end{equation}
where $T_{\mathbf{a},\mathbf{x}}$ is a table of real numbers, while $\beta_{\mathcal{L}}$
is the so-called classical bound of the Bell inequality defined as $\beta_{\mathcal{L}}=\max_{\vec{p}\in\mathcal{L}_{N,m,d}}I$. Analogously, let us denote by $\beta_{\mathcal{Q}}$ and $\beta_{\mathcal{N}}$, respectively, the maximal quantum and nonsignaling values of $I$, i.e., 
\begin{equation}
\beta_{\mathcal{Q}}=\sup_{\vec{p}\in \mathcal{Q}_{N,m,d}}I
,\qquad\beta_{\mathcal{N}}=\max_{\vec{p}\in \mathcal{N}_{N,m,d}}I,
\end{equation}
where the supremum stems from the fact that 
the set of quantum correlations is in general not closed
\cite{Slofstra}.

For most of known Bell inequalities 
$\beta_{\mathcal{L}}<\beta_{\mathcal{Q}}<\beta_{\mathcal{N}}$. In particular, if $\beta_{\mathcal{L}}<\beta_{\mathcal{Q}}$ for some $I$, we call the corresponding Bell inequality proper. Finally, the violation of a Bell inequality by some correlations $\vec{p}$ implies that $\vec{p}$ does not admit the LHV model, in which case it is called nonlocal. 

In the multipartite case ($N > 2$), yet another set of correlations can be considered--the set of Svetlichny correlations \cite{Svetlichny}. To define it formally, let us group the parties $A_1,\ldots,A_N$
into two disjoint subsets $G$ and $\bar{G}=\{A_1,\ldots,A_N\}\setminus G$ such that 
$G\neq \emptyset$. Now, the correlations $\vec{p}$ are called bilocal with respect 
to the bipartition $G|\bar{G}$ if
\begin{equation}
p_{G|\bar{G}}(\mathbf{a}|\mathbf{x})=\sum_{\lambda}p(\lambda)p_{G}(\mathbf{a}_G|\mathbf{x}_{G},\lambda)p_{\bar{G}}(\mathbf{a}_{\bar{G}}|\mathbf{x}_{\bar{G}},\lambda),
\end{equation}
where $\mathbf{a}_{G}$ ($\mathbf{a}_{\bar{G}}$) and $\mathbf{x}_G$ ($\mathbf{x}_{\bar{G}}$) are outcomes and measurement choices corresponding to the observers from $G$ ($\bar{G}$), whereas $p_{G}(\boldsymbol{a}_G|\boldsymbol{x}_{G},\lambda)$ are any probability distributions corresponding to the parties in $G$. We then call $\vec{p}$ bilocal if $p(\mathbf{a}|\mathbf{x})$ can be written as a convex combination of $p_{G|\bar{G}}(\mathbf{a}|\mathbf{x})$ that are bilocal with respect to various bipartitions $G|\bar{G}$. On the other hand, if $\vec{p}$ does not admit the above form, the we call such correlations genuinely multipartite nonlocal. 

In a given Bell scenario $(N,m,d)$ bilocal correlations form a convex set $\mathcal{S}_{N,m,d}$, and for a given Bell expression $I$ we denote by $\beta_\mathcal{S}$ its maximal value over $\mathcal{S}_{N,m,d}$, that is, $\beta_{\mathcal{S}}=\max_{\vec{p}\in\mathcal{S}_{N,m,d}}I$. Violation of $I\leq \beta_{\mathcal{S}}$ by some quantum correlations
$\vec{p}$ indicates that these correlations are genuinely multipartite nonlocal.

Let us stress here that the above definition of bilocality was proven to be inconsistent with the operational interpretation of nonlocality \cite{definition1,definition2}, and to recover consistency it is enough to require that $p_{G}(\mathbf{a}_G|\mathbf{x}_{G},\lambda)$ and $p_{\bar{G}}(\mathbf{a}_{\bar{G}}|\mathbf{x}_{\bar{G}},\lambda)$ obey the no-signaling principle. Nevertheless, it is still of interest to consider the Svetlichny definition of bilocality as any quantum correlations that are genuinely nonlocal according to 
it are also genuinely nonlocal according to the definitions of Ref. \cite{definition1,definition2}. \\

\noindent\textit{A particular multipartite Bell expression.} Let us illustrate the above concepts with the following example of a multipartite Bell expression introduced in Refs. \cite{Aolita,Bancal}, which we state here in the probability picture as
\begin{equation}\label{MultipBI}
I_{\mathrm{ex}}:=\sum_{n=0}^{\lfloor d/2\rfloor-1}\left[\left(1-\frac{2n}{d-1}\right)\left(\mathbbm{P}_n-\mathbbm{Q}_n\right)\right],
\end{equation}
where $\mathbb{P}_n$ and $\mathbb{Q}_n$ are expressions given explicitly by
\begin{equation}\label{Pm}
\mathbbm{P}_n=\sum_{\alpha_1,\ldots,\alpha_{N-1}=1}^{m}\left[ P(X_{\alpha_1,\ldots,\alpha_{N-1}}=n)+P(\overline{X}_{\alpha_1,\ldots,\alpha_{N-1}}=n)\right]
\end{equation}
and
\begin{equation}\label{Qm}
\mathbbm{Q}_n=\sum_{\alpha_1,\ldots,\alpha_{N-1}=1}^{m}\left[ P(X_{\alpha_1,\ldots,\alpha_{N-1}}=-n-1)+P(\overline{X}_{\alpha_1,\ldots,\alpha_{N-1}}=-n-1)\right],
\end{equation}
where $X$ and $\overline{X}$ are linear combinations of the variables $A_{i,x_i}$ defined as
\begin{equation}\label{X}
X_{\alpha_1,\ldots,\alpha_{N-1}}=
A_{1,\alpha_1}+\sum_{j=2}^{N}(-1)^{j-1}A_{j,\alpha_{j-1}+\alpha_j-1}
\end{equation}
and
\begin{equation}\label{X'}
\overline{X}_{\alpha_1,\ldots,\alpha_{N-1}}=
-A_{1,\alpha_1+1}+\sum_{j=2}^{N}(-1)^{j}A_{j,\alpha_{j-1}+\alpha_j-1},
\end{equation}
where we use the convention that $A_{j,m+\gamma}=A_{j,\gamma}+1$ for any $\gamma=1,\ldots,m$ and any $j=1,\ldots,N$, and $\alpha_N:=1$. Moreover, all the equations $X_{\alpha_1,\ldots,\alpha_{N-1}}=k$ or $\overline{X}_{\alpha_1,\ldots,\alpha_{N-1}}=k$ in Eqs. (\ref{Pm}) and (\ref{Qm}) are to be taken \textit{modulo} $d$. While, the maximal classical value of 
$I_{\mathrm{ex}}$ is in general unknown, its Svetlichny bound is straightforward to determine and amounts to $\beta_{\mathcal{S}}^{\mathrm{ex}}=m^{N-2}(m-1)$ \cite{Aolita,Bancal}.

It was proven in Ref. \cite{Aolita} that the Bell inequality (\ref{MultipBI}) 
is violated (but not maximally) by the $N$-qudit GHZ state $\ket{\mathrm{GHZ}_{N,d}}$
together with the following observables
\begin{equation}\label{measurements0}
\mathscr{A}_{1,x}=U_{x}F_d\Omega_dF_d^{\dagger}U_{x}^{\dagger},\qquad
\mathscr{A}_{2,x}=V_{x}F_d^{\dagger}\Omega_dF_dV_{x}^{\dagger},
\end{equation}
for the first two parties, and  
\begin{eqnarray}\label{measurements}
\mathscr{A}_{3,x}&=&W_{x}F_d\Omega_dF_d^{\dagger}W_{x}^{\dagger}\nonumber\\
&\vdots&\nonumber\\
\mathscr{A}_{N-1,x}&=&
\left\{
\begin{array}{ll}
W_{x}F_d\Omega_dF_d^{\dagger}W_{x}^{\dagger}, & \,N\,\,\mathrm{even}\\[1ex]
W_{x}^{\dagger}F_d^{\dagger}\Omega_dF_d W_{x}, & \,N\,\,\mathrm{odd}\\
\end{array}
\right.\nonumber\\
\mathscr{A}_{N,x}&=&\left\{
\begin{array}{ll}
W_{x}^{\dagger}F_d^{\dagger}\Omega_dF_d W_{x}, & \,N\,\,\mathrm{even}\\[1ex]
W_{x}F_d\Omega_dF_d^{\dagger}W_{x}^{\dagger}, & \,N\,\,\mathrm{odd}\\
\end{array}
\right.
\end{eqnarray}
for the remaining parties, where  
\begin{equation}
F_d=\frac{1}{\sqrt{d}}\sum_{i,j=0}^{d-1}\omega^{ij}\ket{i}\!\bra{j},\qquad
\Omega_d=\mathrm{diag}[1,\omega,\ldots,\omega^{d-1}]
\end{equation}
and
\begin{equation}
U_{x}=\sum_{j=0}^{d-1}\omega^{-j \gamma_{m}(x)}\proj{j},\qquad
V_x=\sum_{j=0}^{d-1}\omega^{j \zeta_{m}(x)}\proj{j},\qquad
W_x=\sum_{j=0}^{d-1}\omega^{-j \theta_{m}(x)}\proj{j},\qquad
\end{equation}
with $\gamma_{m}(x)=(x-1/2)/m$, $\zeta_{m}(x)=x/m$, and $\theta_{m}(x)=(x-1)/m$ for $x=1,\ldots,m$. 
Notice that for the case $N=m=2$, these reproduce the optimal CGLMP observables \cite{CGLMP}
and we will use them to construct our Bell inequalities. Moreover, for these observables 
and the state $\ket{\mathrm{GHZ}_{N,d}}$, all the probabilities appearing in 
both $\mathbbm{P}_n$ and $\mathbbm{Q}_n$ [cf. Eqs. (\ref{Pm}) and (\ref{Qm})], that is, 
\begin{equation}
P(X_{\alpha_1,\ldots,\alpha_{N-1}}=n)\qquad \mathrm{and}\qquad
P(\overline{X}_{\alpha_1,\ldots,\alpha_{N-1}}=n)
\label{eq:ToBeProvedInTheAppendix}
\end{equation}
are independent of the choice of $\alpha_1,\ldots,\alpha_{N-1}$ and are equal for any $n=0,\ldots,d-1$, that is, $P(X_{\alpha_1,\ldots,\alpha_{N-1}}=n)=P(\overline{X}_{\alpha_1,\ldots,\alpha_{N-1}}=n)$ (see Appendix for the proof). We will later exploit these properties in our construction of Bell inequalities.\\

\noindent\textit{Generalized correlators.} For further purposes we notice that correlations (\ref{collection}) can also be equivalently represented by correlators instead of conditional probabilities. However, due to the fact that here we consider an arbitrary number of outcomes, we need to appeal to generalized correlators. These are in general complex numbers that are defined through the $N$-dimensional 
discrete Fourier transform of the probabilities $p(\mathbf{a}|\mathbf{x})$ according to the following formula 
\begin{equation}\label{ComplCorr}
\left\langle \mathscr{A}_{1,x_1}^{(k_1)}\ldots \mathscr{A}_{N,x_N}^{(k_N)}\right\rangle=\sum_{\mathbf{k}}\omega^{\mathbf{a}\cdot \mathbf{k}}\,p(\mathbf{a}|\mathbf{x}),
\end{equation}
where the sum is over all $N$-tuples $\mathbf{k}=(k_1,\ldots, k_N)$ with each $k_i=0,\ldots,d-1$ and $\omega=\mathrm{exp}(2\pi \mathbbm{i}/d)$ is the $d$th root of unity and $\mathbbm{i}^2+1=0$. 
It is then not difficult to see that $|\langle \mathscr{A}_{1,x_1}^{(k_1)}\ldots \mathscr{A}_{N,k_N}^{(k_N)}\rangle|\leq 1$ for all configurations of $x_1,\ldots,x_N$ and $k_1,\ldots,k_N$. Moreover, if $k_i=0$ for some $i$, then the no-signalling principle (cf. Eq. \ref{eq:ns}) allows one to rewrite.
\begin{equation}
\left\langle \mathscr{A}_{1,x_1}^{(k_1)}\ldots \mathscr{A}_{N,x_N}^{(k_N)}\right\rangle=\left\langle \mathscr{A}_{1,x_1}^{(k_1)}\ldots \mathscr{A}_{i-1,x_{i-1}}^{(k_{i-1})} \mathscr{A}_{i+1,x_{i+1}}^{(k_{i+1})}\ldots \mathscr{A}_{N,x_N}^{(k_N)}\right\rangle.
\end{equation}
In particular, $\langle \mathscr{A}_{1,x_1}^{(0)}\ldots \mathscr{A}_{1,x_N}^{(0)}\rangle=1$ for any sequence of $x_i$. The inverse transformation gives 
\begin{equation}
p(\mathbf{a}|\mathbf{x})=\frac{1}{d^N}\sum_{\boldsymbol{k}}\omega^{-(\mathbf{a}\cdot\mathbf{k})}\left\langle \mathscr{A}_{1,x_1}^{(k_1)} \ldots \mathscr{A}_{N,x_N}^{(k_N)}\right\rangle.
\end{equation}
We remark that we have used a different letter to denote the new variables instead of $A_{i,x_i}$ because the values that these new variables take are from the regular $d$-gon in the complex plane, with vertices $\{1,\ldots,\omega^{d-1}\}$. Note that for $d=2$, the correlators take values from $[-1,1]$. We also note that for $d>2$ these values are not completely independent: for instance, since the conditional probabilities $p(\mathbf{a}|\mathbf{x})$ are real, its discrete Fourier transform will satisfy the relations 
\begin{equation}
\left\langle \mathscr{A}_{1,x_1}^{(k_1)}\ldots \mathscr{A}_{N,x_N}^{(k_N)}\right\rangle = \left\langle \mathscr{A}_{1,x_1}^{(d-k_1)}\ldots \mathscr{A}_{N,x_N}^{(d-k_N)}\right\rangle^*,
\end{equation}
where $z^*$ denotes the complex conjugate of $z$ and the $k_i$'s are taken \textit{modulo} $d$.

In the case of quantum correlations $p(\mathbf{a}|\mathbf{x})$,
$\mathscr{A}_{i,x_i}$ can be seen as the Fourier transform of the measurement operators $M_{i,x_i}^{a_i}$, that is,
\begin{equation}
\mathscr{A}_{i,x}^{(k)}=\sum_{a}\omega^{ak}M_{i,x}^{a}.
\end{equation}
Phrasing alternatively, one can think of the operators $\mathscr{A}_{i,x_i}^{(k_i)}$ with $k_i=0,\ldots,d-1$ as an observable representation of a $d$-outcome measurement $\{M_{i,x_i}^{a_i}\}$ with outcomes labelled by $1,\omega,\ldots,\omega^{d-1}$. 
Let us also notice that if $M_{i,x}^{a}$ correspond to a projective measurement, then 
$\mathscr{A}_{i,x}$ are unitary operators with eigenvalues $1,\ldots,\omega^{d-1}$
(or, equivalently, they satisfy $\mathscr{A}_{i,x}^d=\mathbbm{1}$). Moreover, in such a case
$\mathscr{A}_{i,x}^{(k)}$ are simply $k$th powers of $\mathscr{A}_{i,x}$, i.e., $\mathscr{A}_{i,x}^{(k)}\equiv \mathscr{A}_{i,x_i}^k$.

\section{The construction}
\label{sec:construction}

Let us now move on to our construction of Bell inequalities maximally violated by the GHZ states (\ref{GHZ}) of any local dimension $d$ and any number of parties $N$. To this end we follow the approach introduced in Ref. \cite{our} to derive Bell inequalities maximally violated by the maximally entangled states of two qudits. Moreover, we impose the condition that the maximal quantum values of the Bell inequalities we here derive are by design obtained for the optimal measurements (\ref{measurements}). 

\subsection{The case of three observers ($N=3$)}

To make our considerations more accessible, we first present our construction in the case of  three observers ($N=3$). As already mentioned, we now denote the parties $A$, $B$ and $C$, whereas their outcomes and measurements choices as $a,b,c$ and $x,y,z$, respectively. The departure point of our considerations is the following generalization of the Bell expression in Eq. (\ref{MultipBI}), 
\begin{equation}\label{BellExp}
I_{3,m,d}=\sum_{n=0}^{\lfloor d/2\rfloor-1}
\left(\alpha_n\mathbb{P}_n-\beta_{n}\mathbb{Q}_n\right).
\end{equation}
Here, the variables defined in Eqs. (\ref{X}) and (\ref{X'}) simplify to
\begin{equation}\label{ParticularCase}
X_{\alpha,\beta}=A_{\alpha}-B_{\alpha+\beta-1}+C_{\beta},
\qquad \overline{X}_{\alpha,\beta}=-A_{\alpha+1}+B_{\alpha+\beta-1}-C_{\beta}.
\end{equation}
and therefore $\mathbb{P}_n$ and $\mathbb{Q}_n$ can be written as
\begin{equation}\label{Pm3}
\mathbb{P}_n=\sum_{\alpha,\beta=1}^{m}
\left[P(A_{\alpha}-B_{\alpha+\beta-1}+C_{\beta}=n)+P(B_{\alpha+\beta-1}-A_{\alpha+1}-C_{\beta}=n)\right]
\end{equation}
and
\begin{equation}\label{Qm3}
\mathbb{Q}_n=\sum_{\alpha,\beta=1}^{m}
\left[P(A_{\alpha}-B_{\alpha+\beta-1}+C_{\beta}=-n-1)+P(B_{\alpha+\beta-1}-A_{\alpha+1}-C_{\beta}=-n-1)\right].
\end{equation}
Moreover, $\alpha_n$ and $\beta_n$ are our free parameters that we are going to determine. Notice that for $\alpha_n=\beta_n=[1-2n/(d-1)]$, one recovers the multipartite Bell inequalities (\ref{MultipBI}). Notice also that the reason 
to consider the above generalization of (\ref{MultipBI}), in which the free parameters multiply 
$\mathbbm{P}_n$ and $\mathbbm{Q}_n$ stems from the fact that
for the state $|\mathrm{GHZ}_{N,d}\rangle$ and the optimal 
measurements (\ref{measurements0}) and (\ref{measurements}), which we later use to 
find our inequalities, all probabilities contributing to $\mathbbm{P}_n$ and $\mathbbm{Q}_n$
are equal (see Appendix for the proof).

We now want to exploit these degrees of freedom in order to construct Bell inequalities maximally violated by $\ket{\mathrm{GHZ}_{3,d}}$.
However, in order to fully appreciate the symmetries inherent in Bell 
inequalities and thus significantly simplify our 
considerations we will write the Bell expression (\ref{BellExp})
in terms of complex correlators (\ref{ComplCorr}) instead of probabilities. After some short algebra one finds that 
\begin{equation}\label{prawd}
P(A_{\alpha}-B_{\alpha+\beta-1}+C_{\beta}=\xi)=\frac{1}{d}\sum_{k=0}^{d-1}\omega^{-k\xi}
\left\langle \mathscr{A}_{\alpha}^k\,\mathscr{B}_{\alpha+\beta-1}^{-k}\,\mathscr{C}_{\beta}^k\right\rangle
\end{equation}
and
\begin{equation}\label{prawd}
P(-A_{\alpha+1}+B_{\alpha+\beta-1}-C_{\beta}=\xi)=\frac{1}{d}\sum_{k=0}^{d-1}\omega^{k\xi}
\left\langle \mathscr{A}_{\alpha+1}^k\,\mathscr{B}_{\alpha+\beta-1}^{-k}\,\mathscr{C}_{\beta}^k\right\rangle.
\end{equation}
With the aid of these formulas our Bell expression (\ref{BellExp}) can be conveniently 
rewritten as
\begin{eqnarray}\label{BellExpCC}
I_{3,m,d}&=&
\frac{1}{d}\sum_{\alpha,\beta=1}^{m}\sum_{k=0}^{d-1}
\left\langle \overline{\mathscr{A}}_{\alpha}^{(k)}\,\mathscr{B}_{\alpha+\beta-1}^{-k}\,\mathscr{C}_{\beta}^k\right\rangle,
\end{eqnarray}
where the new variables $\overline{\mathscr{A}}_{\alpha}^{(k)}$ are defined as
\begin{equation}\label{barowane}
\overline{\mathscr{A}}_{\alpha}^{(k)}=a_k\, \mathscr{A}_{\alpha}^k+a_k^*\, \mathscr{A}_{\alpha+1}^k
\end{equation}
with 
\begin{equation}\label{conditions}
a_k=\sum_{n=0}^{\lfloor d/2\rfloor-1}\left[\alpha_n\omega^{-kn}-\beta_n\omega^{k(n+1)}\right].
\end{equation}
Notice here that as, due to the convention, $\mathscr{A}_{m+1}=\omega \mathscr{A}_{1}$, and therefore in the particular case of $\alpha=m$, Eq. (\ref{barowane}) reads
$\overline{\mathscr{A}}_{m}^{(k)}=a_k\, \mathscr{A}_{m}^k+a_k^* \,\omega^{k} \mathscr{A}_{1}^k$. Let us also notice that the term in Eq. (\ref{BellExpCC}) corresponding to $k=0$ is a constant and therefore it is not included in the Bell expression.

Now, to fix our free parameters $\alpha_n$ and $\beta_n$ $(n=0,\ldots,\lfloor d/2\rfloor-1)$ we
require that for the optimal observables (\ref{measurements}) the following conditions 
\begin{equation}\label{system}
\overline{\mathscr{A}}_{\alpha}^{(k)}\otimes \mathscr{B}_{\alpha+\beta-1}^{-k}\otimes \mathscr{C}_{\beta}^k\ket{\mathrm{GHZ}_{3,d}}=\ket{\mathrm{GHZ}_{3,d}}
\end{equation}
are satisfied for all $\alpha,\beta=1,\ldots,m$ and $k=1,\ldots,d-1$. In other words, we want to find such $\alpha_n$ and $\beta_n$ that the resulting
operator $\overline{\mathscr{A}}_{\alpha}^{(k)}\otimes \mathscr{B}_{\alpha+\beta-1}^{-k}\otimes \mathscr{C}_{\beta}^k$ stabilizes the GHZ state $\ket{\mathrm{GHZ}_{3,d}}$, or that the GHZ state is its eigenstate with eigenvalue one. To solve the above equations we need the explicit forms of the $k$th powers of the measurements (\ref{measurements}). After simple algebra one finds that 
\begin{equation}\label{Axl}
\mathscr{A}_x^k=\omega^{-(d-k)\gamma_m(x)}\sum_{n=0}^{k-1}\ket{d-k+n}\!\bra{n}+\omega^{k\,\gamma_m(x)}\sum_{n=k}^{d-1}\ket{n-k}\!\bra{n},
\end{equation}
\begin{equation}\label{Byl}
\mathscr{B}_y^{-k}=(\mathscr{B}_y^k)^{\dagger}=\omega^{(d-k)\zeta_m(y)}\sum_{n=0}^{k-1}\ket{d-k+n}\!\bra{n}+\omega^{k\,\zeta_m(y)}\sum_{n=k}^{d-1}\ket{n-k}\!\bra{n},
\end{equation}
and
\begin{equation}\label{Czl}
\mathscr{C}_z^k=\omega^{-(d-k)\theta_m(z)}\sum_{n=0}^{k-1}\ket{d-k+n}\!\bra{n}+\omega^{k\,\theta_m(z)}\sum_{n=k}^{d-1}\ket{n-k}\!\bra{n}.
\end{equation}
By plugging (\ref{Axl}), (\ref{Byl}) and (\ref{Czl}) into the conditions (\ref{system}), one obtains the following system of linear equations for the coefficients
$a_k$:
\begin{equation}\label{systemak}
\left\{
\begin{array}{l}
a_k\omega^{-k/2m}+a_k^*\omega^{k/2m}=1,\\[1ex]
a_k\omega^{(d-k)/2m}+a_k^*\omega^{-(d-k)/2m}=1,
\end{array}
\right.
\end{equation}
with $k=1,\ldots,\lfloor d/2\rfloor$.
This system can be directly solved, giving 
\begin{equation}\label{Solutionak}
a_k=\frac{\omega^{\frac{2k-d}{4m}}}{2\cos(\pi/2m)}\qquad (k=1,\ldots,\lfloor d/2\rfloor).
\end{equation}
Having explicit form of $a_k$, we can now excavate the coefficients $\alpha_n$ and $\beta_n$ from the system of equations (\ref{conditions}). The latter consists of $\lfloor d/2\rfloor$
equations involving $2\lfloor d/2\rfloor$ variables, meaning that it does not have a unique solution, and the solution will not be real in general. To avoid this problem we equip 
(\ref{conditions}) with $\lfloor d/2\rfloor$ additional equations of the form 
\begin{equation}\label{system2}
\sum_{n=0}^{\lfloor d/2\rfloor-1}\left(\alpha_n\omega^{kn}-\beta_n\omega^{-(n+1)k}\right)=a_k^*
\end{equation}
which will enforce $\alpha_n$ and $\beta_n$ to be real. 
Both sets of equations (\ref{conditions}) and (\ref{system2}) can be wrapped up 
into a single set of the form
\begin{equation}\label{system2}
\sum_{n=0}^{\lfloor d/2\rfloor-1}\left(\alpha_n\omega^{-kn}-\beta_n\omega^{k(n+1)}\right)=c_k,
\end{equation}
in which $c_k=a_k$ for $k=1,\ldots,\lfloor d/2\rfloor$ and $c_k=c_{-k}^*$ for
$k=-\lfloor d/2\rfloor,\ldots,-1$. In order to solve this system we consider 
the cases of odd and even $d$ separately. For odd $d$, (\ref{system2}) has the same
number of equations and variables, and thus we expect a unique solution, which 
after rather straightforward but tedious calculations is found to be
\begin{equation}\label{alphan2}
\alpha_n=\frac{1}{2d}\tan\left(\frac{\pi}{2m}\right)
\left[g_m(n)-g(\lfloor d/2\rfloor)\right]
\end{equation}
and
\begin{equation}\label{betan2}
\beta_n=\frac{1}{2d}\tan\left(\frac{\pi}{2m}\right)
\left[g_m(n+1-1/m)+g_m(\lfloor d/2\rfloor)\right]
\end{equation}
with $n=1,\ldots,\lfloor d/2\rfloor$ and $g_m(x) := \cot\{\pi[x+ 1/(2m)]/d\}$.

Then, in the case of even $d$ the equations for $k=d/2$ and $k=-d/2$
are the same, and therefore the system (\ref{system2}) consists of $d-1$
equations for $d$ variables. The additional variable we then fix in such a way that
the obtained $\alpha_n$ and $\beta_n$ assume the same form as for odd $d$.

\subsection{Generalization to an arbitrary number of parties}

It turns out that the above considerations remain valid 
if one considers an arbitrary number of parties $N$. 

Let us consider the same Bell expression as in 
(\ref{BellExp}), i.e.,
\begin{equation}\label{BellExpProN}
I_{N,m,d}=\sum_{n=0}^{\lfloor d/2 \rfloor-1 }(\alpha_n \mathbbm{P}_n-\beta_n\mathbbm{Q}_n),
\end{equation}
in which now $\mathbbm{P}_n$ and $\mathbbm{Q}_n$
are defined for any $N$ in (\ref{Pm}) and (\ref{Qm}). 
By using the discrete Fourier transform we can rewrite it 
in terms of the complex correlators as
\begin{equation}\label{BellExpNCC}
\widetilde{I}_{N,m,d}=\frac{1}{d}\sum_{\alpha_1,\ldots,\alpha_{N-1}=1}^{m}
\sum_{k=1}^{d-1}\left\langle 
\overline{\mathscr{A}}_{\alpha_1}^{(k)}
\prod_{i=2}^{N} \left(\mathscr{A}_{i,\alpha_{i-1} + \alpha_{i} - 1}\right)^{(-1)^{i-1} k}\right\rangle,
\end{equation}
where 
$\alpha_{N} = 1$. The variables $\overline{\mathscr{A}}_{\alpha_1}^{(k)}$ are, as before, combinations of $\mathscr{A}_{1,\alpha_1}^{k}$ and $\mathscr{A}_{1,\alpha_1+1}^{k}$
given by
\begin{equation}
\overline{\mathscr{A}}_{\alpha_1}^{(k)}=a_k\,\mathscr{A}_{1,\alpha_1}^k+a_k^*\,\mathscr{A}_{1,\alpha_1+1}^k
\end{equation}
for $\alpha_1=1,\ldots,m$, where, as before, $\mathscr{A}_{1,m+1}=\omega \mathscr{A}_{1,1}$. The coefficients $a_k$ are defined in Eq. (\ref{conditions}). Notice also that the Bell expression (\ref{BellExpNCC}) does not contain the term corresponding to $k=0$ as the latter is only a constant that can always be moved to the classical bound; for this reason we changed the notation from $I_{N,m,d}$ to $\widetilde{I}_{N,m,d}$. The values $I_{N,m,d}$ and $\widetilde{I}_{N,m,d}$ are related by the following formula
\begin{equation}
I_{N,m,d}=\widetilde{I}_{N,m,d}+\frac{m^{N-1}}{d}\sum_{n=0}^{\lfloor d/2\rfloor-1}\left(\alpha_n-\beta_n \right),
\end{equation}
and so for a particular choice of the coefficients $\alpha_n$ and $\beta_n$, given the value of our Bell expression in one representation, we can easily compute it in the other representation.

Now, the above form of the Bell expression suggests the conditions one needs to impose on the variables $\alpha_n$ and $\beta_n$ in order to obtain a Bell inequality maximally violated by the $N$-partite GHZ state $\ket{\mathrm{GHZ}_{N,d}}$. Namely, the following system of equations
\begin{equation}\label{warunek}
\overline{\mathscr{A}}_{\alpha_1}^{(k)}\otimes \bigotimes_{i=2}^{N} \left(\mathscr{A}_{i,\alpha_{i-1} + \alpha_{i} - 1}\right)^{(-1)^{i-1} k}\ket{\mathrm{GHZ}_{N,d}}=\ket{\mathrm{GHZ}_{N,d}},
\end{equation}
with the same conventions as above,
should hold for any sequence of $\alpha_i$'s and $k$ with the measurements being 
given in (\ref{measurements}). After some tedious calculations one finds that this 
leads to the same system of equations for $a_k$ as we obtained in the tripartite case (\ref{systemak}); its solution is given in (\ref{Solutionak}). Thus, our Bell inequalities for any number of parties are determined through the same coefficients $\alpha_n$ and $\beta_n$ as in the case $N=3$. 

To summarize, in the probability representation our class of Bell inequalities is given by (\ref{BellExpProN}) with $\alpha_n$ and $\beta_n$ stated explicitly in Eqs. (\ref{alphan2}) and (\ref{betan2}), respectively, while in the correlator representation it is given by (\ref{BellExpNCC}) with $a_k$ being of the form (\ref{Solutionak}). Moreover, for this choice of $\alpha_n$ and $\beta_n$, we have
\begin{equation}\label{S}
S:=\sum_{n=0}^{\lfloor d/2\rfloor-1}\left(\alpha_n-\beta_n\right)=\frac{1}{2}\left\{1-\tan\left(\frac{\pi}{2m}\right)\cot\left[\frac{\pi}{d}\left(\left\lfloor\frac{d}{2}\right\rfloor+\frac{1}{2m}\right)\right]\right\},
\end{equation}
and consequently the Bell expression in the probability and correlator forms are 
related through the following simple formula
\begin{equation}\label{mapping}
I_{N,m,d}=\widetilde{I}_{N,m,d}+(2m^{N-1}/d)S.
\end{equation}

To conclude this part, let us mention the Bell expressions 
obtained here belong to a more general family of Bell expressions
considered in Ref. \cite{Bancal}, some of which independently discovered in Ref. \cite{AGCA}.
However, in these works it is only shown that high-dimensional GHZ states can exhibit fully random and genuinely multipartite quantum correlations. Our method is such that the inequalities are built from the property that the multiqudit GHZ state and given measurements maximally violate it.

%

%

\section{Characterization}
\label{sec:characterization}

Here we characterize our class of Bell inequalities.
We first aim at computing their local bound. As this turns out to be a hard task, we provide the local bound only for the simplest scenarios; recall that in the bipartite case the classical value was computed analytically in Ref. \cite{our}. We then determine their maximal quantum value, showing at the same time that this value is attained by 
the state $\ket{\mathrm{GHZ}_{N,d}}$ and measurements (\ref{measurements0}) and (\ref{measurements}).
We finally obtain the maximal nonsignaling value of $I_{N,m,d}$.

\subsection{Classical and Svetlichny bounds of our inequalities}
Let us begin by noting that our Bell expression $I_{N,m,d}$ can be written in a simpler form as
\begin{equation}\label{NewHope}
I_{N,m,d}=\sum_{n=0}^{d-1}\hat{\alpha}_n\mathbb{P}_n,
\end{equation}
where $\mathbbm{P}_n$ is given in Eq. (\ref{Pm}) and
$\hat{\alpha}_n=\alpha_n$ for $n=0,\ldots,\lfloor d/2\rfloor-1$ and $\hat{\alpha}_n=-\beta_{d-1-n}$ for $n=\lfloor d/2\rfloor,\ldots,d-1$ (notice that in the odd $d$ case $\alpha_{\lfloor d/2\rfloor}=\beta_{\lfloor d/2\rfloor}=0$).

Let us also recall that to compute the maximal classical value $\beta_{\mathcal{L}}$ of our Bell expressions $I_{N,m,d}$ it is enough to maximize the latter over the vertices of $\mathcal{L}_{N,m,d}$, or, in other words, over all deterministic assignments $A_{i,x_i}\in[d]$ for $x_{i}=1,\ldots,m$ and $i=1,\ldots,N$, i.e., 
\begin{equation}\label{ClassicalValue}
\beta_{\mathcal{L}}=\max_{\{A_{i,k}\in [d]\}_{i \in [N], k \in [M]}} I_{N,m,d}, 
\end{equation}
for which $I_{N,m,d}$ given in (\ref{NewHope}) rewrites as  
\begin{equation}
 I_{N,m,d} = \sum_{m=0}^{d-1}\hat\alpha_m \sum_{\alpha_1, \ldots, \alpha_{N-1} = 1}^M \left[\delta(X_{\alpha_{1},\ldots,\alpha_{N-1}},m)+\delta(\overline{X}_{\alpha_{1},\ldots,\alpha_{N-1}},m)\right],
 \label{ExprExpanded}
\end{equation}
where $\delta(\cdot,\cdot)$ denotes the Kronecker delta, whereas the variables $X$ and $\overline{X}$ are defined in Eqs. (\ref{X}) and (\ref{X'}). 

To facilitate the computation of the classical value we want to express the maximum in (\ref{ClassicalValue}) in terms of the $X$ and $\overline{X}$ variables, instead of 
$A_{i,x_i}$. To do this, we need to remove all the linear dependencies between $X_{\alpha_1,\ldots,\alpha_{N-1}}$ and $\overline{X}_{\alpha_1,\ldots,\alpha_{N-1}}$. Let us illustrate what we mean by this with the bipartite case in which the classical value was computed analytically for any $m$ and $d$ in \cite{our}.

Let us then assume that $N=2$ and notice that 
the variables $X_{\alpha}$ and $\overline{X}_{\alpha}$
are related to $A_{\alpha}$ and $B_{\alpha}$ by the following formula
\begin{equation}
 \left(
 \begin{array}{c}
  1\\ \hline X_{1} \\ \vdots \\X_{m} \\ \hline \overline{X}_1\\ \vdots \\ \overline{X}_{m}
 \end{array}
 \right) = H 
 \left(
 \begin{array}{c}
  1\\ \hline A_1 \\ \vdots \\ A_{m}\\ \hline B_1\\ \vdots\\ B_{m}
 \end{array}
 \right),
 \label{eq:HSystem}
\end{equation}
where
\begin{equation}
H = \left(\begin{array}{c|c|c}
  1&0&0\\
  \hline
  0&\mathbbm{1}&-\mathbbm{1}\\
  \hline
  b&-\mathcal{X}&\mathbbm{1}
 \end{array}\right),
\end{equation}
$b = (0,0,\ldots, 0,-1)^T$ and $\mathcal{X} = \sum_{i=0}^{m-1} \ket{i}\bra{i+1}$. In order to find a linearly independent set of $X_{\alpha}$ and $\overline{X}_{\alpha}$, we want to find all linear combinations such that
\begin{equation}
 a + \sum_{\alpha}\left[ g_{\alpha} X_{\alpha} + h_{\alpha} \overline{X}_{\alpha} \right]= 0,
\end{equation}
where $a$, $g_{\alpha}$ and $h_{\alpha}$ are some coefficients to be determined. For this purpose, we want to determine the kernel of $H^T$, which in this case consists of a single vector $(1,1,\ldots,1)^T$. Consequently, we arrive at the condition for the bipartite Bell expression $I_{2,m,d}$ constructed in Ref. \cite{our}, which is
\begin{equation}
 \sum_{\alpha=0}^{m-1} \left[X_{\alpha} + \overline{X}_{\alpha}\right] \equiv -1 \mod d.
 \label{ConditionX}
\end{equation}
As proven in Ref. \cite{our}, this allows to nest the optimization of the classical bound and use a dynamic programming procedure to find it efficiently.
Indeed, one can re-express the optimization in (\ref{ClassicalValue}) in terms of the $X_{\alpha}, \overline{X}_{\alpha}$, but with the condition stemming from (\ref{ConditionX}).

This allows us to write
\begin{equation}
\beta_{\mathcal{L}}=\max_{\{X_{\alpha}, \overline{X}_{\alpha}\in [d]\ : \sum_\alpha X_\alpha + \overline{X}_\alpha \equiv -1 \mod d\}} I_{N,m,d},
\label{eq:SmartMaximization}
\end{equation}
and eliminate the variables in the optimization successively thanks to the form of (\ref{ExprExpanded}). The corresponding classical bound of $I_{2,m,d}$ is then found to be \cite{our}:
\begin{equation}\label{Classical2}
\beta_{\mathcal{L}}^{2,m,d} = \frac{1}{2d} \tan\left(\frac{\pi}{2m}\right) \left[ (2m - 1) g(0) - 
g\left(1-\frac{1}{m}\right) - 2m g\left( \left\lfloor \frac{d}{2} \right\rfloor \right) \right].
\end{equation}

In what follows we attempt to formalize the procedure to find the classical bound of the inequality for a larger system size. Here we outline the steps, but several obstacles arise in the multipartite case that we currently do not see how to overcome.

Let us consider an arbitrary number of observers $N$. Likewise, we wish to find all $a, g_{\boldsymbol{\alpha}}, h_{\boldsymbol{\alpha}} \in \mathbbm{Z}_d$ such that
\begin{equation}\label{dupa}
 a + \sum_{\boldsymbol{\alpha}} \left[g_{\boldsymbol{\alpha}} X_{\boldsymbol{\alpha}} + h_{\boldsymbol{\alpha}}\overline{X}_{\boldsymbol{\alpha}}\right] \equiv 0 \mod d.
\end{equation}
Note however, that now $\boldsymbol{\alpha}$ is a vector and therefore one expects the solution set to be multidimensional. In other words, $H$ is a square matrix only for $N=2$. If $N>2$ we shall have on the left hand side of (\ref{eq:HSystem}) a $(1+2m^{N-1})$-component vector of $X_{\boldsymbol{\alpha}}$'s and $\overline{X}_{\boldsymbol \alpha}$'s, whereas on the left hand side we shall have a ($1+mN$)-component vector of observables. Therefore, the number of terms in the kernel of $H^T$ will grow exponentially with $N$, implying there will be an exponential number of non-trivial constraints in the maximization analogous to (\ref{eq:SmartMaximization}).
In order to find them, using Eqs. (\ref{X}) and (\ref{X'}), the above equation can be expanded as
\begin{equation}
 a + \sum_{\boldsymbol{\alpha}}\left[g_{\boldsymbol{\alpha}}A_{1,\alpha_1}-h_{\boldsymbol{\alpha}}A_{1,\alpha_1+1}\right]+\sum_{i=2}^N(-1)^{i-1}\sum_{\boldsymbol{\alpha}}\left[g_{\boldsymbol{\alpha}} - h_{\boldsymbol{\alpha}} \right] A_{i,\alpha_{i-1} + \alpha_i - 1} \equiv 0 \mod d.
\end{equation}
This gives a set of equations that make the coefficients in front of $A_{N,\alpha_{N-1}}$ congruent to $0\,\, \mathrm{mod}\,\, d$:
\begin{equation}
 \sum_{\boldsymbol{\alpha}:{\alpha_{N-1}=k}}\left[g_{\boldsymbol{\alpha}} - h_{\boldsymbol{\alpha}}\right] \equiv 0 \mod d \qquad (1 \leq k \leq m). 
\end{equation}
Similarly, for $1 < j < N$, we make the coefficient in front of $A_{j,k}$ congruent to $0\,\, \mathrm{mod}\,\, d$:
\begin{equation}
 \sum_{\boldsymbol{\alpha}:{\alpha_{j-1}+\alpha_{j}-1 = k}} \left[g_{\boldsymbol{\alpha}} - h_{\boldsymbol{\alpha}}\right] \equiv 0 \mod d \qquad (1 \leq k \leq m).
\end{equation}
Then, the coefficient that multiplies $A_{1,\alpha_1}$ is
\begin{equation}
 \sum_{\boldsymbol{\alpha}:{\alpha_1 = k}} \left[g_{\alpha_1=k,\boldsymbol{\alpha'}} - h_{\alpha_1=k-1,\boldsymbol{\alpha'}}\right] \equiv 0 \mod d \qquad (1 \leq k \leq m),
\end{equation}
where $\boldsymbol{\alpha}'\equiv\alpha_2,\ldots,\alpha_{N-1}$. We have also an equation for the constant term in Eq. (\ref{dupa}). Here we have to take into consideration that $A_{i,m-1 + k} = A_{i,k} + 1$ for $k>0$ and any $i=1,\ldots,N$. 

In order to appropriately do the substitutions of the $X_{\boldsymbol{\alpha}}$'s and $\overline{X}_{\boldsymbol \alpha}$'s, one needs to perform Gauss elimination on a basis of this kernel. However, we note that the equations (\ref{dupa}) are over $\mathbbm{Z}_d$, which is a field only of $d$ is a prime number. Therefore, for some values of $d$, inverses may not exist, and it can be much more complicated to obtain a good basis of $\ker{H}^T$. This was not a problem for $N=2$ as $\ker H^T$ was generated simply by a vector of ones. In addition, it is unclear how to later exploit the properties of the $g$ function that were used in \cite{our} in order to find an analytical form of $\beta{\mathcal L}^{N,m,d}$ for $N>2$.

Thus, even though the procedure above helps to slightly reduce the complexity of the optimization in some particular cases with a few number of particles, we resort to numerics in the general case.
As pointed out in Theorem \ref{class_thm}, we have that $\beta_{\mathcal{L}}^{3,2,d} = \beta_{\mathcal{S}}^{3,2,d}$ and its value is thus given by expression (\ref{Classical2}). For completeness, we include the classical bound values in the simplest Bell scenarios
in tables \ref{TabelkaI} and \ref{TabII}.

\begin{table}[h]
\label{TabelkaI}
\begin{center}
\renewcommand{\arraystretch}{1.5}
\begin{tabular}{|c|c|c|c|}
\cline{1-4}
\multicolumn{4}{|c|}{$N=3$, $M=2$}\\
\hline
$d$ & $2$ & $3$ & $4$ \\ \hline
$\beta_{\mathcal{L}}^{3,2,d} = 
\beta_{\mathcal{S}}^{3,2,d}$ & 4.2426 & 3.0416 & 3.5953 \\
 \hline
\multicolumn{4}{|c|}{$N=3$, $M=3$}\\
\hline
 \multicolumn{1}{|c|}{$\beta_{\mathcal{L}}^{3,3,d}$} & $\frac{13}{\sqrt{3}}$ & $\tfrac{1}{6\sqrt{3}}[13 \cot\left(\tfrac{\pi}{18}\right) - 17\tan\left(\tfrac{\pi}{9}\right) - 4\tan\left(\tfrac{2\pi}{9}\right)]$ & $\frac{-10 + 17\sqrt{2} + 14\sqrt{6}}{4\sqrt{3}}$ \\ 
& = 7.5056 & = 6.1760 & = 6.9765 \\\hline
\multicolumn{1}{|c|}{$\beta_{\mathcal{S}}^{3,3,d}$} & 8.6603 & 7.3132 & 8.1115 \\ \hline
\end{tabular}
\renewcommand{\arraystretch}{1}
\end{center}
\caption{Maximal classical and Svetlichny values of $I_{N,m,d}$ for $N=3$, $M=2,3$ and $d=2,3,4$.}
\end{table}

\begin{table}
\label{TabII}
\begin{center}
\renewcommand{\arraystretch}{1.5}
\begin{tabular}{|c|c|c|c|}
\cline{1-4}
\multicolumn{4}{|c|}{$N=4$, $M=2$}\\\hline
$d$ & $d=2$ & $d=3$ & $d=4$ \\ \hline
\multicolumn{1}{|c|}{$\beta_{\mathcal{L}}^{4,2,d}$} & $\frac{5}{2}[\cot\left(\tfrac{\pi}{8}\right) + \tan\left(\tfrac{\pi}{8}\right)]$ & $\frac{10}{\sqrt{3}} + \frac{5}{6}(-3 + \sqrt{3})$ & $\frac{1}{8}[10 \cot\left(\frac{\pi}{16}\right) - 5\cot\left(\frac{3\pi}{16}\right) $ \\ 
\multicolumn{1}{|c|}{} &  & & $+ 16\tan\left(\frac{\pi}{16}\right) + \tan\left(\frac{3\pi}{16}\right)]$ \\
\multicolumn{1}{|c|}{} & $= 7.0711$ & $= 4.7169$ & $= 5.8301$\\ \hline
\multicolumn{1}{|c|}{$\beta_{\mathcal{S}}^{4,2,d}$} & 8.4853 & 6.0829 & 7.1905 \\ \hline
\multicolumn{4}{|c|}{$N=4$, $M=3$}\\\hline
\multicolumn{1}{|c|}{$\beta_{\mathcal{L}}^{4,3,d}$} & $\frac{35}{\sqrt{3}}$ & $\tfrac{7}{6\sqrt{3}}[5 \cot\left(\frac{\pi}{18}\right) - 7\tan\left(\frac{\pi}{9}\right) - 2\tan\left(\frac{2\pi}{9}\right)]$ &  \\ 
& = 20.2073 & = 16.2537 & \\\hline
\multicolumn{1}{|c|}{$\beta_{\mathcal{S}}^{4,3,d}$} & 25.9808 & 21.9394 & 24.3345\\ \hline
\end{tabular}
\renewcommand{\arraystretch}{1}
\end{center}
\caption{Maximal classical and Svetlichny values of $I_{N,m,d}$ for $N=4$, $M=2,3$ and $d=2,3,4$.}
\end{table}

On the other hand, it is quite direct to obtain an upper bound on the 
maximal value of $I_{N,m,d}$ over the Svetlichny correlations. Precisely, 
the results of Ref. \cite{Bancal} allow us to state the following theorem.

\begin{thm}\label{class_thm}The Svetlichny bound of $I_{N,m,d}$ is bounded from above as 
$\beta_{\mathcal{S}}^{N,m,d} \leq  m^{N-2} \beta_{\mathcal{L}}^{2,m,d}$, where $\beta_{\mathcal{L}}^{2,m,d}$ is the classical bound of the bipartite Bell inequality given explicitly in (\ref{Classical2}).
\end{thm}
\begin{proof}The proof is given in Ref. \cite{Bancal}.
\end{proof}
It is worth mentioning that for the case $N=3$ and $m=2$ the bound $\beta_{\mathcal{S}}^{3,2,d}$ is also saturated by fully product probability distribution (see also Table \ref{TabelkaI})
\begin{equation}
p_{\mathrm{loc}}(\boldsymbol{a}|\boldsymbol{x})=p_{A}(a|x)p_{B}(b|y)p_{C}(c|z)
\end{equation}
such that $p_{A}(0|x)=p_{A}(0|y)=p_{C}(0|z)=1$ for all $x,y,z$. So, $\beta_{\mathcal{S}}^{3,2,d}$ is also the classical bound of the corresponding Bell inequality. In general, however, the classical and Svetlichny bounds differ.

\subsection{Quantum and nonsignaling bounds}

Let us now move on to the quantum and nonsignaling bounds.

\begin{thm}\label{quantum_thm} The maximal quantum value of $\widetilde{I}_{N,m,d}$ is $\widetilde{\beta}_Q^{N,m,d}=m^{N-1}(d-1)/d$.
\end{thm}
\begin{proof}To prove that $\beta_{Q}^{N,m,d}$ is an upper bound on the maximal quantum value of $I_{N,m,d}$ we can follow the method of Ref. \cite{Bancal}. Here, however, we follow
an alternative approach exploiting the sum-of-squares decompositon of the shifted Bell operator, which might be of use for such applications of nonlocality as self-testing.

To this end, let us consider a Bell operator $\mathcal{B}_{N,m,d}$ constructed from
the Bell expression $\widetilde{I}_{N,d,m}$ with some observables
$\mathscr{A}_{i,x_i}$ (that is, unitary operators such that $\mathscr{A}_{i,x_i}^d=\mathbbm{1}$):
\begin{equation}\label{BellExpNCC}
\mathcal{B}_{N,m,d}:=\frac{1}{d}\sum_{\alpha_1,\ldots,\alpha_{N-1}=1}^M
\sum_{k=1}^{d-1}
\overline{\mathscr{A}}_{\alpha_1}^{(k)}\otimes
\bigotimes_{i=2}^{N} \left(\mathscr{A}_{i,\alpha_{i-1} + \alpha_{i} - 1}\right)^{(-1)^{i-1} k},
\end{equation}
Our aim now is to prove that the operator $\beta_{\mathcal{Q}}^{N,m,d}\mathbbm{1}-\mathcal{B}_{N,m,d}$
is positive semi-definite for arbitrary observables $\mathscr{A}_{i,x_i}$ with the identity operator $\mathbbm{1}$ being of defined on the corresponding Hilbert space. This can be achieved by decomposing $\beta_{\mathcal{Q}}^{N,m,d}\mathbbm{1}-\mathcal{B}_{N,m,d}$ into a sum of squares. To be more precise, let us first consider the simpler case of $m=2$ and introduce the following operators
\begin{equation}
P^{(k)}_{\alpha_1,\ldots,\alpha_{N-1}}=\mathbbm{1}- \overline{\mathscr{A}}_{\alpha_1}^{(k)}\otimes \bigotimes_{i=2}^{N}\left(\mathscr{A}_{i,\alpha_{i-1} + \alpha_{i} - 1}\right)^{(-1)^{i-1} k}.
\end{equation}
Then, by a direct check one finds that the following decomposition
\begin{equation}
\beta_{\mathcal{Q}}^{N,2,d}\mathbbm{1}-\mathcal{B}_{N,2,d}=\frac{1}{2d}\sum_{\alpha_1,\ldots,\alpha_{N-1}=1}^{m}
\sum_{k=1}^{d-1}\left[P^{(k)}_{\alpha_1,\ldots,\alpha_{N-1}}\right]^{\dagger}P^{(k)}_{\alpha_1,\ldots,\alpha_{N-1}}
\end{equation}
holds true.
In the case of arbitrary number of measurements, the above sum of squares needs to be 
slightly modified. Let us introduce the following operators
\begin{equation}
T_{\alpha}^{(k)}=\mu^*_{\alpha,k}\,\mathscr{A}_{1,2}^k+\nu^*_{\alpha,k}\,\mathscr{A}_{1,\alpha+2}^k+\tau_{\alpha,k}\,\mathscr{A}_{1,\alpha+3}^k
\end{equation}
for $\alpha=1,\ldots,m-2$ and $k=1,\ldots,d-1$, where the coefficients 
are defined as
\begin{eqnarray}
\mu_{\alpha,k} & = & \frac{\omega^{(\alpha + 1)(d - 2k)/2m}}{2 \cos (\pi/2m)} \frac{\sin (\pi/m)}{\sqrt{\sin (\pi \alpha/m) \sin \left[\pi (\alpha +1)/m\right]}}, \nonumber \\
\nu_{\alpha,k} & = & -\frac{\omega^{(d - 2k)/2m}}{2 \cos (\pi/2m)} \sqrt{\frac{\sin \left[\pi (\alpha +1)/m\right]}{\sin (\pi \alpha/m)}}, \nonumber \\
\tau_{\alpha,k} & = & \frac{1}{2 \cos (\pi/2m)} \sqrt{\frac{\sin (\pi \alpha/m)}{ \sin \left[\pi (\alpha +1)/m\right]}}=-\frac{\omega^{(d - 2k)/2m}}{4\cos^2(\pi/2m)}\nu_{\alpha,k}^{-1},
\label{coeffabc}
\end{eqnarray}
for $i = 1, \ldots,m-3$ and $k = 1, \ldots, d-1$, while for $i=m-2$ and $k=1,\ldots,d-1$ they are given by 
\begin{eqnarray}
\mu_{m-2,k} & = & - \frac{\omega^{-(d - 2k)/2m}}{2\sqrt{2} \cos (\pi/2m)\sqrt{\cos (\pi/m)}},\nonumber \\
\nu_{m-2,k} & = & - \frac{\omega^k \omega^{(d - 2k)/2m}}{2\sqrt{2} \cos (\pi/2m)\sqrt{\cos (\pi/m)}},\nonumber \\
\tau_{m-2,k} & = & \frac{\sqrt{\cos (\pi/m)}}{\sqrt{2}\cos (\pi/2m)}.
\label{coeffm2}
\end{eqnarray}
Then, the sum of squares is given by
\begin{equation}
\beta_{\mathcal{Q}}^{N,m,d}\mathbbm{1}-\mathcal{B}_{N,m,d}=\frac{1}{2d}\sum_{\alpha_1,\ldots,\alpha_{N-1}=1}^{m}
\sum_{k=1}^{d-1}\left[P^{(k)}_{\alpha_1,\ldots,\alpha_{N-1}}\right]^{\dagger}P^{(k)}_{\alpha_1,\ldots,\alpha_{N-1}}+\frac{m^{N-2}}{2d}\sum_{\alpha=1}^{m-2}\sum_{k=1}^{d-1}
\left[T_{\alpha}^{(k)}\right]^{\dag}T_{\alpha}^{(k)}.
\end{equation}
To conclude the proof, let us notice that for the state $\ket{\mathrm{GHZ}_{N,d}}$ and the measurements (\ref{measurements0}) and (\ref{measurements}) the value of $\widetilde{I}_{N,m,d}$ is clearly $m^{N-1}(d-1)/d$, which follows from the fact that for this realisation
each correlator in $\widetilde{I}_{N,m,d}$ assumes value one (cf. Eq. \ref{warunek}).
$\beta_{\mathcal{Q}}^{N,m,d}$ is thus the maximal quantum value of $\widetilde{I}_{N,m,d}$.
\end{proof}

\begin{thm}\label{ns_thm}The maximal nonsignaling value of $I_{N,m,d}$ equals its algebraic bound and it is given by $\beta_{\mathcal{N}}^{N,m,d}=  2m^{N-1}\alpha_0$.
\end{thm}
\begin{proof}To prove this statement we use the form of $I_{N,m,d}$ given in Eq. (\ref{NewHope}).
As shown in Ref. \cite{our} (see the supplemental material), $\alpha_0\geq \alpha_n$ for any $0 \leq n \leq d-1$, and consequently one obtains the following bound by putting all the terms in $\mathbb{P}_0$ equal to one:
\begin{equation}\label{albound}
I_{N,m,d}\leq 2m^{N-1}\alpha_0.
\end{equation}
Now, there exists a nonsignaling probability distribution for which this inequality is saturated. For the first portion of measurement choices it is defined as
\begin{eqnarray}
&&\hspace{-1.5cm}p(a_1,\ldots, a_N|\alpha_1,\alpha_1+\alpha_2-1,\ldots,\alpha_{N-2}+\alpha_{N-1}-1,\alpha_{N-1})\nonumber\\
&&=
\left\{
\begin{array}{ll}
\displaystyle\frac{1}{d^{N-1}}, & \displaystyle\sum_{i=1}^{N}(-1)^{i-1} a_i
=f(\alpha_1,\ldots,\alpha_{N-1})\\[4ex]
0, & \mathrm{otherwise}
\end{array}
\right.,
\end{eqnarray}
for $\alpha_i=0,\ldots,m-1$ and $i=1,\ldots,N-1$, where
the function $f$ is defined as
\begin{equation}
f(\alpha_1, \ldots, \alpha_{N-1}) = \sum_{i = 1}^{N-2} (-1)^{i-1} H(\alpha_i + \alpha_{i+1} - m - 1).
\end{equation}
with $H$ being the discrete Heaviside step function, defined as $H(x) = 1$ if $x \geq 0$ and $H(x) = 0$ otherwise. This function $f$ is introduced to take into account the convention $A_{i,m+k}=A_{i,k}+1$, which modifies the condition defining the probabilities in the Bell expression. Indeed, looking at the expression (\ref{Pm}), one sees that if for all $i=1,\ldots,N-2$, $\alpha_i+\alpha_{i+1}-1\leq m-1$, then $f=0$, but if for some $j$'s, $\alpha_j+\alpha_{j+1}-1>m-1$, then $f$ could be different than $0$. 

Then, for the other portion of measurement choices it is defined as
\begin{eqnarray}
&&\hspace{-1.5cm}p(a_1,\ldots, a_N|\alpha_1+1,\alpha_1+\alpha_2-1,\ldots,\alpha_{N-2}+\alpha_{N-1}-1,\alpha_{N-1})\nonumber\\
&&=
\left\{
\begin{array}{ll}
\displaystyle\frac{1}{d^{N-1}}, & \displaystyle\sum_{i=1}^{N}(-1)^{i-1}a_i=\widetilde{f}(\alpha_1,\ldots,\alpha_{N-1})\\[4ex]
0, & \mathrm{otherwise}
\end{array}
\right.,
\end{eqnarray}
where 
the function $\widetilde{f}$ is defined in the same way as $f$, but also takes into account that $\alpha_1+1$ can be larger than $m-1$. Thus
\begin{equation}
\widetilde{f}(\alpha_1, \ldots, \alpha_{N-1}) = -H(\alpha_1 +1- m) + f(\alpha_1, \ldots, \alpha_{N-1}).
\end{equation}
For all the remaining choices of measurements we define
\begin{equation}
p(a_1,\ldots, a_N|\alpha_1,\ldots, \alpha_N)=\frac{1}{d^N}.
\end{equation}
Let us now recall the no-signalling principle for many parties. For the distribution of elements $p(a_1, \ldots, a_N | x_1, \ldots, x_N)$,  the marginal $p(a_{i_1}, \ldots, a_{i_k} | x_{i_1}, \ldots, x_{i_k})$ for any subset $\{i_1, \ldots, i_k\}$ of the $N$ parties should be independent of the measurement settings of the remaining $N-k$ parties:
\begin{equation}
p(a_{i_1}, \ldots, a_{i_k} | x_{1}, \ldots, x_{N}) = p(a_{i_1}, \ldots, a_{i_k} | x_{i_1}, \ldots, x_{i_k}).
\end{equation}
It is not difficult to verify that the distribution presented above obeys the no-signaling principle. Tracing out a single subsystem one always obtains a maximally random probability distribution.
\end{proof}

Let us notice that in Theorem \ref{quantum_thm} we compute the maximal quantum value of our Bell expression in the correlator representation (\ref{BellExpNCC}), whereas in Theorem 
\ref{ns_thm} we compute the maximal nonsignalling value in the probability picture. 
To obtain these values in the other picture one can use Eq. (\ref{mapping}).

Let us also notice that both these values are related to the same 
values of the bipartite SATWAP Bell inequality by a factor $m^{N-2}$. 
Thus, the results about the relative scaling of these bounds 
from Ref. \cite{our} holds for many parties. In particular 
$\lim_{d\to\infty}\beta_{\mathcal{NS}}^{N,m,d}/\beta_{\mathcal{Q}}^{N,m,d}=(2m/\pi)\tan(\pi/2m)$, and so the separation between the maximal nonsignaling and quantum values becomes smaller for larger $m$.

On the other hand, the classical value $\beta_{\mathcal{L}}^{N,m,d}$ does not seem
to obey $\beta_{\mathcal{L}}^{N,m,d}=m^{N-2}\beta_{\mathcal{L}}^{2,m,d}$ (cf Tabs. \ref{TabelkaI} and \ref{TabII}), and therefore one can expect that the behaviour
of the ratio $\beta_{\mathcal{Q}}^{N,m,d}/\beta_{\mathcal{L}}^{N,m,d}$
will exhibit a behaviour for large $d$ or $m$ different from $\beta_{\mathcal{Q}}^{2,m,d}/\beta_{\mathcal{L}}^{2,m,d}$ (cf. Ref. \cite{our}).

\subsection{Special cases}

Let us here briefly discuss the form of our Bell expressions in the special cases of $d=2$ and any $m\geq 2$, and $m=2$ and any $d\geq 2$. In the first one, Eq. (\ref{BellExpProN}) simplifies to 
\begin{equation}
I_{N,m,2}=\alpha_0\mathbbm{P}_0,
\end{equation}
where $\alpha_0=1/[2\cos(\pi/2m)]$ (notice also that $\beta_0=0$).
Then, in the correlator picture there is a single number $a_1=\alpha_0=1/[2\cos(\pi/2m)]$ and therefore 
\begin{equation}
\overline{\mathscr{A}}_{\alpha_1}^{(1)}=a_1(\mathscr{A}_{\alpha_1}+\mathscr{A}_{\alpha_1+1}),
\end{equation}
where $\mathscr{A}_{1,m+1}=-\mathscr{A}_{1,1}$. Then, the Bell inequality
in the correlator picture can be written as
\begin{equation}
\widetilde{I}_{N,m,2}=\frac{a_1}{2}\sum_{\alpha_1,\ldots,\alpha_{N-1}=1}^{m}
\left[\left\langle \mathscr{A}_{1,\alpha_1}\prod_{i=2}^{N} \mathscr{A}_{i,\alpha_{i-1} + \alpha_{i} - 1}\right\rangle+\left\langle \mathscr{A}_{1,\alpha_1+1}\prod_{i=2}^{N} \mathscr{A}_{i,\alpha_{i-1} + \alpha_{i} - 1}\right\rangle\right]\leq \beta_{\mathcal{L}}^{N,m,2},
\end{equation}
where $\alpha_N=1$. This a generalization of the bipartite
chained Bell inequalities \cite{chained} to the multipartite scenario (see also Ref. \cite{AGCA} for an extension in a similar spirit, albeit in which the GHZ state does not yield the maximal violation in general). 
In fact, for $N=2$, after dividing by $a_1/2$, one obtains 
\begin{equation}
\widetilde{I}_{2,m,2}=\sum_{\alpha=1}^m\left[\langle \mathscr{A}_{\alpha}\mathscr{B}_{\alpha}\rangle+\langle \mathscr{A}_{\alpha+1}\mathscr{B}_{\alpha}\rangle\right]\leq 2(m-1),
\end{equation}
where $\mathscr{A}_{m+1}=-\mathscr{A}_{1}$, which is the chained Bell inequality \cite{chained}.

In the case of $m=2$ and any $d$, $I_{N,2,d}$ is given in 
Eq. (\ref{BellExpProN}) with the coefficients in the probability picture simplifying to
\begin{equation}
\alpha_k=\frac{1}{2d}\left[g_2(k)+(-1)^d\tan\left(\frac{\pi}{4d}\right)\right],\qquad
\beta_k=\frac{1}{2d}\left[g_2(k+1/2)-(-1)^d\tan\left(\frac{\pi}{4d}\right)\right],
\end{equation}
whereas in the correlator picture to $a_k=\omega^{(2k-8)/d}/\sqrt{2}$.

\section{Classes of inequalities tailored to partially entangled states}
\label{sec:generalizations}

In this section, we investigate whether Bell inequalities of the form (\ref{BellExpProN}) can be tailored to give a class of inequalities maximally violated by partially entangled states. This is a natural question to ask, given that in the case of two parties, the CGLMP \cite{CGLMP} and the SATWAP \cite{our} inequalities are maximally violated by different entangled states, and are both of the form (\ref{BellExpProN}) with different coefficients $\alpha_m$ and $\beta_m$. We first present the case $N=2$, $d=3$ which was already studied in \cite{Science}, and then consider extensions to new cases $N=3,4$ and $d=4$.

\subsection{$N = m = 2$, $d=3$}
In this special case, Eq. (\ref{BellExpProN}) gives a class of Bell inequalities involving two parameters
$\alpha_0\mathbbm{P}_0-\beta_0\mathbbm{Q}_0\leq C$, where $C$ is the maximal classical value.
However, we can always divide the whole expression by one of them, say $\alpha_0$ (provided that it is positive), reducing the number of free parameters to one.
As a result we obtain the following class of Bell inequalities
\begin{eqnarray}\label{class}
J_{2,2,3}(\xi)&\!\!\!:=\!\!\!&P(A_1=B_1)+P(A_2=B_2)+P(A_1=B_2-1)+P(A_2=B_1)\nonumber\\
&&-\xi[P(A_1=B_1-1)+P(A_2=B_2-1)+P(A_1=B_2)+P(A_2=B_1+1)]\leq C_3(\xi),\nonumber\\
\end{eqnarray}
parametrized by a single parameter $\xi$, defined 
as  $\xi = \beta_0/\alpha_0$. It turns out
that the classical bound of these inequalities can be easily found by maximizing $J_{2,2,3}(\xi$) over all local deterministic strategies, which gives 
\begin{equation}
C_3(\xi)=\left\{
\begin{array}{ll}
-4\xi, & \mathrm{if}\,\,\, \xi\leq -1,\\
3-\xi, & \mathrm{if}\,\,\, -1\leq \xi\leq 1,\\
2, & \mathrm{if}\,\,\, \xi\geq 1.
\end{array}
\right.
\end{equation}
Moreover, numerical tests using the Navascu\'es-Pironio-Ac\'in (NPA) hierarchy \cite{npa2007}
indicate that for $\xi \leq -1$, the Bell inequality (\ref{class}) is trivial, meaning that
its maximal quantum violation equals its classical bound. Consequently, in what follows
we will concentrate on the case $\xi>-1$. It is not difficult to see that for $\xi=1$ the class (\ref{class}) reproduces the well-known CGLMP Bell inequality \cite{CGLMP}, which is known to be maximally violated by the partially entangled state \cite{Durt}:
\begin{equation}\label{partially}
\ket{\psi_\gamma}=\frac{1}{\sqrt{2+\gamma^2}}(\ket{00}+\gamma\ket{11}+\ket{22})
\end{equation}
with $\gamma=(\sqrt{11}-\sqrt{3})/2$, whereas for $\xi=(\sqrt{3}-1)/2$ it gives the SATWAP Bell inequality. In both cases the optimal CGLMP observables (expression (\ref{measurements0}) for $N=m=2$) are used.

The question we want to answer now is whether by changing $\xi$ between the above two values we can obtain Bell inequalities maximally violated by partially entangled states (\ref{partially}) for various values of $\gamma$. To answer this question let us first take
the observables (\ref{measurements0}) and compute the value of the Bell expression for the state (\ref{class}). This gives us the following function of $\xi$ and $\gamma$:
\begin{equation}\label{expression}
\mathcal{J}(\xi,\gamma)=\frac{4}{3}\frac{3+\gamma(2\sqrt{3}+\gamma-\xi\gamma)}{2+\gamma^2}.
\end{equation}
To find its maximal value for a fixed $\xi$, we need to satisfy the following condition
$\partial \mathcal{J}(\xi,\gamma)/\partial \gamma=0$, which is equivalent to finding the root of a second degree polynomial in $\gamma$. The maximal value of $\mathcal{J}(\xi,\gamma)$ is found to be at 
$\gamma_{+}(\xi)=[(4\xi^2+4\xi+25)^{1/2}-2\xi-1]/2\sqrt{3}$,
and it is given by 
\begin{equation}\label{expression2}
\mathcal{J}_{\max}(\xi)=\frac{1}{3}\left[5 - 2 \xi + \sqrt{25 + 4 (\xi+1) \xi}\right].
\end{equation}

Of course, the above derivation is not a proof that, for a given $\xi$, $\mathcal{J}_{\mathrm{max}}(\xi)$ is the maximal quantum violation of the Bell inequality (\ref{class}), however, based on our numerical study we conjecture this to be the case. Notice first that for $\xi=1$ and $\xi=(\sqrt{3}-1)/2$, the expression
(\ref{expression2}) reproduces the maximal quantum violations of the CGLMP and SATWAP Bell inequalities, respectively. Then, we have tested our conjecture for other values of $\xi$ by using the NPA  hierarchy, which we implemented using the Yalmip toolbox \cite{yalmip} and the SeDuMi solver \cite{sedumi} in Matlab. The NPA hierarchy provides outer approximations to the quantum set of correlations, and for a given Bell inequality, it allows one to find an upper bound on the maximal quantum violation of the Bell esxpression.  We employed this technique for values of $\xi\in[-0.99,100]$ with the step 0.01, and for all these values of $\xi$ the value obtained agrees with (\ref{expression2}) up to solver precision $~10^{-8}$, which is a strong implication that it is the maximal quantum violation of the corresponding inequality. Note that for $\xi\in[-0.99,42]$, the level $1+AB$ of the hierarchy was sufficient, while for $\xi\in[42,100]$ we used the level $2$, except for a small amount of values in the interval $[85,100]$ for which the level $2 + AAB$ was necessary. 

In \cite{Science}, we also showed how this class of inequalities can be used to self-test partially entangled states using the method of Ref. \cite{selftesting1}.

\subsection{Extension to $N = 3,4$}
The extension of the last section to more parties turns out to be straightforward. We follow the same procedure: we start from $(\ref{BellExpProN})$ for $N=3,4$, $m=2$, $d=3$, and divide it by one the parameters so that we obtain a one-parameter class of Bell expressions
\begin{eqnarray}
J_{N,2,3} (\xi) = \mathbbm{P}_0- \xi \mathbbm{Q}_0, 
\end{eqnarray}
with $N=3,4$ (let us notice here that both $\mathbbm{P}_0$
and $\mathbbm{Q}_0$ defined in Eqs. (\ref{Pm}) and (\ref{Qm}) depend on $N$). We can compute the classical bound of these expressions, obtaining
 \begin{eqnarray}
C_{3,2,3}(\xi)=\left\{
\begin{array}{ll}
-8\xi, & \mathrm{if}\,\,\, \xi\leq -1,\\
2(3-\xi), & \mathrm{if}\,\,\, -1\leq \xi\leq 1,\\
4, & \mathrm{if}\,\,\, \xi\geq 1
\end{array}
\right. 
\end{eqnarray}
and
\begin{eqnarray}
C_{4,2,3}(\xi)=\left\{
\begin{array}{ll}
-16\xi, & \mathrm{if}\,\,\, \xi\leq -10/11,\\
10-5\xi, & \mathrm{if}\,\,\, -10/11 \leq \xi\leq 2/5,\\
8, & \mathrm{if}\,\,\, \xi\geq 2/5.
\end{array}
\right.
\end{eqnarray}
 
Let us now consider the following partially entangled GHZ states
\begin{equation}
|\text{GHZ}^{(N)}_{\gamma}\rangle = \frac{1}{\sqrt{2 + \gamma^2}} (|0\rangle^{\otimes N} + \gamma |1\rangle^{\otimes N} + |2\rangle^{\otimes N}).
\end{equation}
As in the previous subsection, we compute the values of $\mathcal{J}^{(N)}(\xi,\gamma)$ $(N=3,4)$ for the corresponding partially entangled GHZ states and the measurements (\ref{measurements0}) and (\ref{measurements}), and then we solve $\partial \mathcal{J}^{(N)}(\xi,\gamma)/\partial \gamma=0$ to obtain the optimal
\begin{equation}\label{gamma_multi}
\gamma^{(3)}(\xi)= \gamma^{(4)}(\xi) = \frac{\sqrt{4\xi^2+4\xi+25}-2\xi-1}{2\sqrt{3}},
\end{equation}
which is the same value as for $N=2$. Substituting (\ref{gamma_multi}) into the values of the Bell expressions, one obtains
\begin{equation}
\mathcal{J}^{(3)}_{\max}(\xi) = 2 (1 + 2\xi + \sqrt{25 + 4(\xi + 1)\xi}), \label{maxi3}
\end{equation}
and $\mathcal{J}^{(4)}_{\max}(\xi)=2\mathcal{J}^{(3)}_{\max}(\xi)$.
We conjecture that they are the maximal quantum violations of $J_{3,3,2} (\xi) $ and $J_{4,3,2} (\xi)$, respectively.

To support this conjecture, we use the NPA hierarchy. With the change of scenario, it takes more time to solve each SDP, so we do not check as many values of $\xi$ as in the section above. For $N=3$, we checked values of $\xi \in [-1,5]$ with step $0.1$ and found that the gap was of order $10^{-7}$ or lower.  For $N=4$, we checked values of $\xi \in [-1,2]$ with step $0.5$ and found that the gap was of order $10^{-8}$ or lower.

\section{Conclusion}
\label{sec:conclusion}
In this work, we have designed a new family of Bell inequalities in the most general scenario involving $m$ $d$-outcome measurements per observer such that the GHZ state of $N$ qudits maximally violates it, for any $N$ and $d$. Whereas the natural approach towards finding new, useful, families of Bell inequalities is typically based on exploiting the geometry of the set of local correlations (i.e., trying to characterize the facets of the so-called local polytope), tailoring Bell inequalities to quantum states of interest has proven to be a much more successful approach towards the certification of quantum properties of these states \cite{our,Jed,Science}. This shift of approach is perhaps surprising, as CHSH inequality, the simplest non-trivial Bell inequality, possesses many of the properties one desires to certify in practice (e.g. self-testing the singlet state of two qubits). However, one should have in mind that there is no \textit{a-priori} reason why these desirable properties of CHSH should be inherited in more complicated Bell scenarios, simply because local hidden variable theories have nothing to do with quantum theory.

Therefore, in order to certify, in a device-independent manner, properties of quantum states of interest, the roadmap we here suggest looks like a much more promising approach: one generates a probability distribution in the set of quantum correlations that is extremal and exposed (i.e., is the unique maximizer of a Bell functional) and certifies this maximal violation by giving a sum of squares decomposition of the Bell operator. We note that, although in our approach the difficulty of computing the maximal quantum bound is removed, by construction, now finding the classical bound of such inequality becomes in general a non-trivial task. In our work we have computed it exactly in the simplest Bell scenarios with the aid of numerics. Observe, however, that in order to certify the quantum properties of interest, it is not necessary to compute exactly the classical bound of the inequality, and a relaxation of its value (e.g. given by an outer approximation of the local polytope e.g. \cite{FlavioPRX, ThetaBodies}) will suffice.

We have also shown that our method can be adapted to other families of GHZ-like states, in analogy to non-maximally entangled states of two qudits \cite{our,Science}. This is possible because our method is fully analytical, thus enabling us to further introduce analytical parameters and obtain the result only using elementary differential geometry techniques.

Finally, the inequalities we here present can be tested with currently-available technology. In the bipartite case \cite{our}, their application had already been shown in an integrated photonics device \cite{Science}.

\section{Acknowledgments}
\label{sec:acknos}
J. T. thanks the Alexander von Humboldt foundation for support. 
R. A. acknowledges the support from the Foundation for Polish Science through the First TEAM project (First TEAM/2017-4/31) co-financed by the European Union under the European Regional Development Fund.
A. A. acknowledges support from the ERC CoG QITBOX, the AXA Chair in Quantum Information Science, the Spanish MINECO (QIBEQI FIS2016-80773-P and Severo Ochoa SEV-2015-0522), Fundaci\'o Cellex, Generalitat de Catalunya (SGR 1381 and CERCA Programme).

\appendix

\section{Proof that the correlators in Eq. (\ref{eq:ToBeProvedInTheAppendix}) are equal}

Here our aim is to show that for the measurements (\ref{measurements0}) and (\ref{measurements}) and the GHZ state, all the probabilities in Eqs. (\ref{Pm}) and (\ref{Qm}), that is, 
\begin{equation}
P(X_{\alpha_1,\ldots,\alpha_{N-1}}=k)\qquad\mathrm{and}\qquad
P(\overline{X}_{\alpha_1,\ldots,\alpha_{N-1}}=k)
\end{equation}
with $k=0,\ldots,d-1$ (recall that the equalities in the arguments of these probabilities are modulo $d$), are independent of the choice of $\alpha_1,\ldots,\alpha_{N-1}$
and are equal for any $k=0,\ldots,d-1$. 

To this end, let us first notice that the eigenvectors of the observables 
(\ref{measurements0}) and (\ref{measurements}) can be written as
(cf. \cite{CGLMP} and \cite{Aolita})
\begin{eqnarray}
\ket{a_{1},x_1}&=&\frac{1}{\sqrt{d}}\sum_{q=0}^{d-1}\omega^{q\left[a_1-\gamma_{m}(x_1)\right]},\\
\ket{a_2,x_2}&=&\frac{1}{\sqrt{d}}\sum_{q=0}^{d-1}\omega^{-q\left[a_2-\zeta_{m}(x_2)\right]},
\end{eqnarray}
for the first two observers, and
\begin{equation}
\ket{a_i,x_i}=\frac{1}{\sqrt{d}}\sum_{q=0}^{d-1}\omega^{(-1)^{N+1}q\left[a_i-\theta_{m}(x_i)\right]}
\end{equation}
for $i=3,\ldots,N$. Recall that in these formulas $\gamma_m(x)=(x-1/2)/m$, $\zeta_m(x)=x/m$,
and $\theta_m=(x-1)/m$. This means that the joint probability of obtaining $a_i$ 
by party $A_i$ upon measuring the observable $\mathscr{A}_{i,x_i}$ on the state $|\mathrm{GHZ}_{N,d}\rangle$ reads
\begin{eqnarray}\label{probs}
&&\hspace{-0.5cm}p(a_1,\ldots,a_N|x_1,\ldots,x_N)\nonumber\\
&&=\left(\frac{1}{\sqrt{d}}\right)^{N+1}
\left|\sum_{q=0}^{d-1}\exp\left\{\frac{2\pi\mathbbm{i}}{d}q\left[\sum_{i=1}^{N}(-1)^{i+1}a_i-\left(\gamma_m(x_1)-\zeta_m(x_2)+\sum_{i=3}^{N}(-1)^{i+1}\theta_m(x_i)\right)\right]\right\}\right|^2\nonumber\\
&&=\left(\frac{1}{\sqrt{d}}\right)^{N+1}
\left|\sum_{q=0}^{d-1}\exp\left\{\frac{2\pi\mathbbm{i}}{d}q\left[\sum_{i=1}^{N}(-1)^{i+1}a_i-\frac{1}{m}\sum_{i=1}^{N}(-1)^{i+1}x_i+\Lambda_{N,m}\right]\right\}\right|^2,
\end{eqnarray}
where we denoted $\Lambda_{N,m}=[2-(-1)^N]/(2m)$. 

Let us then concentrate on $P(X_{\alpha_1,\ldots,\alpha_{N-1}}=k)$ and consider first the case when $\alpha_{i-1}+\alpha_i-1\leq m$ for $i=2,\ldots,N-1$. Substituting then 
$x_1=\alpha_1$, $x_i=\alpha_{i-1}+\alpha_i-1$ with $i=2,\ldots,N-1$ and $x_N=\alpha_N$ 
in Eq. (\ref{probs}), we can write
\begin{eqnarray}\label{formula}
P(X_{\alpha_1,\ldots,\alpha_{N-1}}=k)&=&\left(\frac{1}{\sqrt{d}}\right)^{N+1}
\sum_{\substack{a_1,\ldots,a_N=0, \\\sum_{i}(-1)^{i+1}a_i=k}}^{d-1}\left|\sum_{q=0}^{d-1}\exp\left\{\frac{2\pi\mathbbm{i}}{d}q\left[\sum_{i=1}^{N}(-1)^{i+1}a_i+\Lambda'_{m}\right]\right\}\right|^2\nonumber\\
&=&\left(\frac{1}{\sqrt{d}}\right)^{N+1}
\sum_{\substack{a_1,\ldots,a_N=0, \\\sum_{i}(-1)^{i+1}a_i=k}}^{d-1}\left|\sum_{q=0}^{d-1}\exp\left[\frac{2\pi\mathbbm{i}}{d}q\left(k+\Lambda'_{m}\right)\right]\right|^2,\nonumber\\
&=&\left(\frac{1}{\sqrt{d}}\right)^{N+1} d^{N-1}\left|\sum_{q=0}^{d-1}\exp\left[\frac{2\pi\mathbbm{i}}{d}q\left(k+\Lambda'_{m}\right)\right]\right|^2,
\end{eqnarray}
where $\Lambda'_{m}=\Lambda_{N,m}-[1+(-1)^{N+1}]/(2m)=1/(2m)$, and to get the last line we used the fact that the expression under the first sum does not depend on $a_i$ and that due to the constraint there are $d^{N-1}$ elements in that sum. Clearly, the expression appearing on the right-hand side of the above formula does not depend on $\alpha_i$. 

Let us then assume that for some $i$, $\alpha_{i-1}+\alpha_i-1>m$. For all such $i$'s 
we use the convention that $A_{i,m+\gamma}=A_{i,\gamma}+1$, which implies that 
$P(X_{\alpha_1,\ldots,\alpha_{N-1}}=k)=P(X'_{\alpha_1,\ldots,\alpha_{N-1}}=k+f(\alpha))$
where in $X'$,  $\alpha_{i-1}+\alpha_i-1$ are replaced by $\alpha_{i-1}+\alpha_i-1-m$
for all those $i$'s for which $\alpha_{i-1}+\alpha_i-1>m$.
\begin{eqnarray}\label{formula2}
P(X'_{\alpha_1,\ldots,\alpha_{N-1}}=k+f(\mathbf{\alpha}))&=&\left(\frac{1}{\sqrt{d}}\right)^{N+1}\nonumber\\
&&\times\hspace{-1cm}
\sum_{\substack{a_1,\ldots,a_N=0, \\\sum_{i}(-1)^{i+1}a_i=k+f(\alpha)}}^{d-1}\left|\sum_{q=0}^{d-1}\exp\left\{\frac{2\pi\mathbbm{i}}{d}q\left[\sum_{i=1}^{N}(-1)^{i+1}a_i-f(\alpha)+\Lambda'\right]\right\}\right|^2\nonumber\\
&=&\left(\frac{1}{\sqrt{d}}\right)^{N+1}\sum_{\substack{a_1,\ldots,a_N=0, \\\sum_{i}(-1)^{i+1}a_i=k+f(\alpha)}}^{d-1}\left|\sum_{q=0}^{d-1}\exp\left[\frac{2\pi\mathbbm{i}}{d}q\left(k+\Lambda'\right)\right]\right|^2\nonumber\\
&=&\left(\frac{1}{\sqrt{d}}\right)^{N+1}d^{N-1}\left|\sum_{q=0}^{d-1}\exp\left[\frac{2\pi\mathbbm{i}}{d}q\left(k+\Lambda'\right)\right]\right|^2,
\end{eqnarray}
where to the third line follows from the fact that, as above, the expression under
the sum does not depend on $a_i$'s and that there are $d^{N-1}$ terms in that sum.
Again, this formula does not depend on $\alpha_i$'s and equals the one in 
Eq. (\ref{formula}).

In a similay way one proceeds with 
$P(\overline{X}_{\alpha_1,\ldots,\alpha_{N-1}}=k)$. There are, however, two differences with respect to the previous case: first, in Eq. (\ref{probs}) one substitutes $x_1=\alpha_1+1$ instead of $x_1=\alpha_1$, second, the condition for outcomes in Eq. (\ref{formula}) modifies to $\sum_{i}(-1)^{i+1}a_i=-k$. Nevertheless, after some calculations one finds that 
$P(\overline{X}_{\alpha_1,\ldots,\alpha_{N-1}}=k)=P(X_{\alpha_1,\ldots,\alpha_{N-1}}=k)$ for any $k$ and any choice of measurements $\alpha_i$, which is what we wanted to show.
%
%
%
%
%
%

\begin{thebibliography}{100}

\bibitem{Bell}J. S. Bell, Physics \textbf{1}, 195 (1964).

\bibitem{randomness1}S. Pironio \textit{et al.}, Nature \textbf{464}, 1021 (2010).

\bibitem{dimwit}N. Brunner, S. Pironio, A. Ac\'in, N. Gisin, A. A. M\' ethot, and
V. Scarani, Phys. Rev. Lett. \textbf{100}, 210503 (2010).

\bibitem{MayersYao}D. Mayers and A. Yao, \textit{Quantum cryptography with
imperfect apparatus}, in Proceedings of 39th Annual Symposium
(FOCS) (1998), 503.

\bibitem{selftesting}M. McKague, T. H. Yang, and V. Scarani, 
J. Phys. A: Math. Theor. \textbf{45}, 455304 (2012); T. H. Yang and M. Navascu\'es, 
Phys. Rev. A \textbf{87}, 050102(R) (2013); A. Coladangelo, K. T. Goh, and V. Scarani, 
Nat. Commun. \textbf{8}, 15485 (2017).

\bibitem{CHSH}J. F. Clauser, M. A. Horne, A. Shimony, and R. A. Holt,
Phys. Rev. Lett. \textbf{23}, 880 (1969).

\bibitem{tilted}A. Ac\'in, S. Massar, S. Pironio, Phys. Rev. Lett. \textbf{108}, 100402 (2012).

\bibitem{CGLMP}D. Collins, N. Gisin, N. Linden, S. Massar, S. Popescu, Phys.
Rev. Lett. \textbf{88}, 040404 (2002).

\bibitem{BM05}H. Buhrman and S. Massar, Phys. Rev. A \textbf{72}, 052103 (2005).


\bibitem{BKP}J. Barrett, A. Kent, S. Pironio, Phys. Rev. Lett. \textbf{97},
170409 (2006).

\bibitem{VP08}T. V\'ertesi and K. P\' al, Phys. Rev. A \textbf{77}, 042106 (2008).

\bibitem{Ji08}S.-W. Ji \textit{et al.}, Phys. Rev. A \textbf{78}, 052108 (2008).

\bibitem{Liang09}Y.-C. Liang, C.-W. Lim and D.-L. Deng, Phys. Rev. A \textbf{80}, 
052116 (2009).



\bibitem{Wiesiek}W. Laskowski, T. Paterek, M. \.Zukowski, \v{C}. Brukner, 
Phys. Rev. Lett. \textbf{93}, 200401 (2004).

\bibitem{GraphBell}O. G\"uhne, G. T\'oth, P. Hyllus, and H. J. Briegel, 
Phys. Rev. Lett. \textbf{95}, 120405 (2005). 

\bibitem{SLK06}W. Son, J. Lee, and M. S. Kim, Phys. Rev. Lett. \textbf{96}, 060406 (2006).

\bibitem{Lim10}J. Lim \textit{et al.}, New J. Phys. \textbf{12}, 103012 (2010).

\bibitem{Aolita}L. Aolita \textit{et al.}, Phys. Rev. Lett. \textbf{108}, 100401 (2012).






\bibitem{our}A. Salavrakos, R. Augusiak, J. Tura, P. Wittek, A. Ac\'in and S. Pironio, Phys. Rev. Lett. \textbf{119}, 040402 (2017).

\bibitem{Jed}J. Kaniewski, I. \v{S}upi\'c, J. Tura, F. Baccari, A. Salavrakos, R. Augusiak, 
\textit{Maximal nonlocality from maximal entanglement and mutually unbiased bases, and self-testing of two-qutrit quantum systems}, arXiv:1807.03332.

\bibitem{Andrea}A. Coladangelo, Phys. Rev. A \textbf{98}, 052115 (2018).

\bibitem{Mermin}N. D. Mermin, Phys. Rev. Lett. \textbf{65}, 1838 (1990).

\bibitem{LCL07}S.-W. Lee, Y.-W. Cheong, and J. Lee, Phys. Rev. A \textbf{76}, 032108 (2007).





\bibitem{ourGraph}F. Baccari, R. Augusiak, I. \v{S}upi\'c, J. Tura, A. Ac\'in, 
\textit{Scalable Bell inequalities for graph states and robust self-testing}, 
arXiv:1812.10428.


\bibitem{Science} J. Wang \textit{et al.}, Science \textbf{360}, 285-291 (2018).

\bibitem{Bancal} J. D. Bancal, C. Branciard, N. Brunner, N. Gisin, and Y.-C. Liang, J. Phys. A \textbf{45},125301 (2012).



\bibitem{PR} S. Popescu, and D. Rohrlich, Found. Phys. \textbf{24}, 379 (1994).

\bibitem{Svetlichny} G. Svetlichny, Phys. Rev. D \textbf{35}, 3066 (1987).

\bibitem{Slofstra}W. Slofstra, \textit{The set of quantum correlations is not closed}, arXiv:1703.08618.


\bibitem{definition1}R. Gallego, L. E. Würflinger, A. Ac\'in, and M. Navascu\'es,
Phys. Rev. Lett. \textbf{109}, 070401 (2012).

\bibitem{definition2}J.-D. Bancal, J. Barrett, N. Gisin, and S. Pironio, Phys. Rev.
A \textbf{88}, 014102 (2013).


\bibitem{chained}P. A. Pearle, Phys. Rev. D \textbf{2}, 1418 (1970);
S. L. Braunstein and C. Caves, Ann. Phys. (N.Y.) \textbf{202}, 22 (1990).


\bibitem{npa2007} M. Navascu\'{e}s, S. Pironio, and A. Ac\'{i}n, Phys. Rev. Lett. \textbf{98}, 010401 (2007); New J. Phys. {\bf 10}, 073013 (2008).

\bibitem{Durt}A. Ac\'in, T. Durt, N. Gisin and J. I. Latorre, Phys. Rev. A \textbf{65}, 052325 (2002).





\bibitem{yalmip} J. L\"{o}fberg, In Proceedings of the CASCD Conference (2004).

\bibitem{sedumi} J. F. Sturm, Optimization methods and software, 11-12 (1999).

\bibitem{selftesting1}T. H. Yang \textit{et al.}, Phys. Rev. Lett. \textbf{113}, 040401 (2014).

%





%

%
%

%


\bibitem{FlavioPRX}F. Baccari, D. Cavalcanti, P. Wittek, and A. Ac\'in,
Phys. Rev. X \textbf{7}, 021042 (2017). 

\bibitem{ThetaBodies} M. Fadel, and J. Tura,
Phys. Rev. Lett. \textbf{119}, 230402 (2017). 


\bibitem{AGCA} L. Aolita, R. Gallego, A. Cabello, and A. Ac\'in,
Phys. Rev. Lett. \textbf{108}, 100401 (2012).






\end{thebibliography}
\end{document}